\documentclass[a4paper,11pt]{article}
\pdfoutput=1 

\usepackage{jheppub} 

\usepackage{subfigure}
\usepackage{amsthm}
\usepackage{amsmath}

\usepackage[T1]{fontenc} 

\graphicspath{{./figures/}}
\usepackage{bbold}
\usepackage{tikz}
\usepackage{tkz-euclide}
\usepackage{adjustbox}
\usepackage{pst-solides3d}
\usetikzlibrary{matrix}
\usepackage{tikz-3dplot}
\usepackage[all,cmtip]{xy}

\theoremstyle{plain}
\newtheorem{thm}{Theorem}[section]
\newtheorem{lem}[thm]{Lemma}
\newtheorem{prop}[thm]{Proposition}
\theoremstyle{definition}
\newtheorem{defn}[thm]{Definition}
\newtheorem{exmp}[thm]{Example}
\theoremstyle{remark}

\usepackage{braket}

\def\hom{\text{hom}(C,G)}
\def\homA{\text{hom}(C_A,G)}
\def\homtA{\text{hom}(C_{\tilde{A}},G)}

\def\homdA{\text{hom}(C_{\partial A}, G)}
\def\homB{\text{hom}(C_B,G)}
\def\Hom{\text{Hom}}
\def\Z{\mathbb{Z}}
\def\H{\mathcal{H}}
\def\C{\mathbb{C}}

\def\ker{\text{ker}}
\def\im{\text{Im}}

\title{ \textbf{Topological Entanglement Entropy in \(d\)-dimensions for Abelian Higher Gauge Theories}}

\author[a,1]{J. P. Ibieta-Jimenez,\note{Corresponding author.}}
\author[a]{M. Petrucci,}
\author[a]{L. N. Queiroz Xavier,}
\author[a]{P. Teotonio-Sobrinho}

\affiliation[a]{Departamento de Fisica Matematica, Universidade de Sao Paulo\\ Rua do Matao Travessa R 187, CEP 05508-090, Sao Paulo, Brazil.}

\emailAdd{pibieta@if.usp.br}
\emailAdd{marzia@if.usp.br}
\emailAdd{lucasnix@if.usp.br}
\emailAdd{teotonio@if.usp.br}

\abstract{We compute the topological entanglement entropy for a large set of lattice models in $d$-dimensions. It is well known that many such quantum systems can be constructed out of lattice gauge models. For dimensionality higher than two, there are generalizations going beyond gauge theories, which are called higher gauge theories and rely on higher-order generalizations of groups. Our main concern is a large class of $d$-dimensional quantum systems derived from Abelian higher gauge theories.
In this paper, we derive a general formula for the bipartition entanglement entropy for this class of models, and from it we extract both the area law and the sub-leading terms, which explicitly depend on the topology of the entangling surface. 
We show that the entanglement entropy $S_A$ in a sub-region $A$ is proportional to $\log(GSD_{\tilde{A}})$, where \(GSD_{\tilde{A}}\) is the ground state degeneracy of a particular restriction of the full model to \(A\).
The quantity $GSD_{\tilde{A}}$ can be further divided into a contribution that scales with the size of the boundary $\partial A$ and a term which depends on the topology of $\partial A$. There is also a topological contribution coming from $A$ itself, that may be non-zero when $A$ has a non-trivial homology. 
We present some examples and discuss how the topology of $A$ affects the topological entropy. 
Our formalism allows us to do most of the calculation for arbitrary dimension $d$. The result is in agreement with entanglement calculations for known topological models. 
 }

\begin{document} 
\maketitle
\flushbottom
\section{\label{sec:intro}Introduction}

The concept of entanglement entropy in quantum many-body systems is increasingly gaining relevance for both the quantum information and the condensed matter theory communities. In the latter case, the interest comes from applying ideas of quantum information that could provide new tools for the study of quantum many-body systems and, in consequence, to deepen the understanding of their quantum phases. In particular, questions about the scaling of this entropy with the system size appear to be relevant as an indicator for quantum entanglement. Of particular interest is the scaling of entanglement entropy for ground states of gapped systems, since they often follow an \emph{area law} \cite{Srednicki, eisert10, Kitaev2006}. More precisely, if we consider a distinguished sub-region \(A\) of the total system, the scaling of entanglement entropy is linear with the boundary of the region, \(\partial A\). See \cite{eisert10} for a detailed account on the occurrence of area laws for the entanglement entropy of quantum systems.  

The growing interest on the study of entanglement entropy in quantum many-body systems arises from different points of view. For example, a source of interest in the scaling area law of entanglement entropy comes from asking whether a quantum many-body system can be simulated by a classical computer. The scaling of entanglement entropy specifies how well a given many-body quantum state can be approximated by a matrix-product state or a PEPS \cite{Verstraete06}. More importantly for the purposes of this work, the \emph{topological entanglement entropy} \cite{Kitaev2006, Levin06} arises as an interesting probe for \emph{topological order} \cite{Wen90, Wen04, Nussinov09} in quantum states. The entanglement entropy calculated in the ground states of topologically ordered states follows an \emph{area law} plus a universal correction indicating the presence of long-range entanglement.

Topological phases of matter are usually characterized by exhibiting long-range entanglement and non-local order parameters such as the ground state degeneracy ($GSD$) and topological spins. In addition, entanglement entropy turns out to be a good measure of the presence of topological order \cite{Aguado08, Furukawa07}. Details about the connection between topological entanglement entropy and topological order are for example exposed in \cite{Nussinov09,Castelnovo2007, Levin06}. For two dimensional topological phases, the scaling of the entanglement entropy presents a constant term \cite{Levin06, Kitaev2006} that corresponds to the \emph{topological entropy}. This result is examined in detail for the Toric Code in \cite{hamma05, Castelnovo2007, Hart2018, Orus2014}. In this paper, we show that this extends to Abelian higher gauge models in all dimensions. Indeed if the model is topological, it presents a non-zero topological entropy given in terms of the higher cohomology groups of both the bulk and the boundary $\partial A$ of the subregion $A$.

In general, two dimensional topological order and topological entropy are relatively well understood. The same cannot be said about higher dimensional topological phases. 
Our interest here is to shed some light on the main features of topological entanglement entropy in dimensions higher than two. To do so, we restrict ourselves to study Abelian lattice models that come from higher gauge theories, such as the ones studied by \cite{Kapustin13,Kapustin14,Bullivant16,Bullivant17,higher,Zhu18,Delcamp18,Delcamp18towards,Delcamp19}. Some of these models are topological. They present a ground state degeneracy that depends on the topology of the underlying manifold, and extended excitations generalizing anyons to higher dimensions. The models presented here can be interpreted as higher gauge generalizations of the Toric Code. This simplification allows us to work in arbitrary dimensions and to understand in detail how these models depend on the topology of the underlying manifold. In the present work we look at the topological entropy and how it depends on the topology of the subregion $A$.

The definition of entanglement entropy is straightforward: we consider a bipartition of the system into a sub-region \(A\) and its complement \((B)\). Let \(\rho\) be the density matrix of the ground space state, defined in the whole lattice. Then, \(\rho_A = \text{Tr}_B (\rho)\) is the reduced density matrix, obtained by tracing out the contribution from region \((B)\). The entanglement entropy is then defined as the \emph{von Neumann} entropy of the reduced density matrix, namely,
\begin{equation}
S_A:= - \text{Tr}\left(\rho_A \log \rho_A\right).
\end{equation}
In a gapped phase, the entanglement entropy is expected to satisfy an \emph{area law} as the leading term. The topological information is contained in subleading terms and, in general, it is not easy to extract. Several prescriptions \cite{Kitaev2006,Levin06} were constructed in order to extract the topological correction to the entanglement entropy in two dimensional gapped systems. These prescriptions have been generalized \cite{Castelnovo08,Grover11} for \(d=3\) and, consequently, used to successfully obtain the entanglement entropy of fracton models \cite{Ma2018,Schmitz2018}. 

In this paper, we study the entanglement entropy of $n$-dimensional Abelian higher gauge theories, all at once, without the need to adapt the procedure for each dimension. This can be achieved by using the language of homological algebra, in which higher gauge theories are naturally described as shown in \cite{higher}. The way the entanglement entropy is obtained, in essence, relies on the fact that the lattice models are constructed as \emph{stabilizer codes} \cite{Gottesman97}. Similar calculations for stabilizer codes are performed in \cite{Fattal2004,hamma05,Hamma2005,He2018,Zou16}. The result obtained for the entanglement entropy relates this quantity to the ground state degeneracy of an associated model, which we write as
\begin{align}
S_A = \log \left(GSD_{\tilde{A}}\right),
\end{align}
where $GSD_{\tilde{A}}$ is the ground state degeneracy of a particular restriction of the original model, to be defined in section \ref{sec:bipart}.
We show that this result allows us to write $S_A$ as a sum of two terms. The first one $S_{\partial A}$ is in agreement with the \emph{area laws} depending on the geometry of the boundary ${\partial A}$, while the second term $S_{\text{Topo}}$ is explicitly topological, depending on the topology of both the bulk of the region $A$ and its boundary $\partial A$. This result holds for \textit{any} higher gauge theory in the sense of \cite{higher} and for \textit{any dimension} $d$. We explicitly calculate examples where the topological term depends not only on the Betti numbers of the subregion $A$ and its boundary $\partial A$, but also on more exotic properties.
\par The paper is structured as follows. In section \ref{sec:ExampleAHGT} we begin by reviewing Abelian higher gauge theories in detail, this is done by giving explicitly an example. Next, in section \ref{sec:higher} we show how these models are easily described in the language of homological algebras, as described in detail in \cite{higher}. In section \ref{sec:entent} the calculation of the entanglement entropy is performed, and we show how to extract both the \text{area law} and the topological entropy from it. In section \ref{sec:Examples} we apply the results of section \ref{sec:entent} to examples in 2D and 3D. We end the paper with some final remarks in section \ref{sec:Remarks}.

\section{Models from Abelian Higher Gauge Theories. The warm-up example}\label{sec:ExampleAHGT}

Before we define our class of models in full generality, it is convenient to first present an examples in 3 dimensions. It will be described in the usual way as a 
many-body system defined on a lattice. The Hilbert space ${\mathcal H}$ consists of
quantum states attached to elements of the lattice such as vertices, links, and 
plaquettes. A Hamiltonian $H$ acting on ${\mathcal H}$
completes the picture. 

In section \ref{sec:higher} we will introduce a formalism that allows us to describe all models in this class in a unified way for all dimensions; this is made possible by employing a few constructions coming from homological algebra. As far as the present section, we don't need to be concerned with all the homological details, but we will point out some of the chain complexes that will be part of the construction presented in section \ref{sec:higher}. It is worth mentioning again that the example exhibited in this section is merely illustrative. We choose to present the model for a $3$-dimensional cubic lattice to be more intuitive. Through our formalism, to be presented in \S\ref{sec:higher}, these models can be defined and studied in \textit{any dimension} $d$ for more general cell decompositions, including simplicial complexes. 

\subsection{0,1,2-Gauge in 3D}\label{ex:12}

This model is a generalization of an ordinary gauge theory. In addition to the
states localized on the links and labeled by an Abelian group $G_1$, we also have states associated with plaquettes that are labelled by another Abelian group $G_2$, and states on vertices associated to the Abelia group $G_0$. Here we use the additive notation for the group operation. When $G_0, G_2$ are trivial we recover the Quantum Double model defined in $3$-dimensions. When $G_0=0$, this class of models
corresponds to an Abelian version of the models constructed in \cite{Bullivant16, Bullivant17},
where 2-groups were considered. We will leave the discussion about the relation between our formalism and the one based on $2$-groups (or more generally, $n$-groups) to section \ref{sec:higher}. Here we are working with a particular case of Abelian 2-groups, where the only data we need is expressed by the group homomorphism $\partial_2^G:G_2\to G_1$.

For simplicity, we consider a $3$-dimensional space $\Sigma$ discretized by a cubic lattice. Let $K_0, K_1, K_2$ and $K_3$ be the sets of vertices, links, faces and cubes, respectively. For each vertex $v\in K_0$ we
have a Hilbert space ${\H}_v$ with orthornomal basis $\left\{ \ket{h}, h\in G_0\right\}$. For each link $l\in K_1$ we
have a Hilbert space ${\H}_l$ with orthornomal basis $\left\{ \ket{g}, g\in G_1\right\}$. 
For each plaquette $p\in K_2$ we have a Hilbert space $\H_p$ with orthornomal
basis $\left\{ \ket{\alpha} , \alpha\in G_2\right\}$. 
We call this kind of model a $0,1,2$-gauge theory to indicate that there
are quantum states associated with sets $K_0$ (vertices), $K_1$ (links) and $K_2$ (plaquettes).
The total Hilbert space is the tensor product over all local Hilbert spaces, namely
\begin{equation}\label{HilSp-12}
 \H=\bigotimes_{v\in K_0}\H_v\bigotimes_{l\in K_1}\H_l\bigotimes_{p\in K_2}\H_p .	
\end{equation}
It is convenient to represent a basis element of $\mathcal{H}$ as functions \((f_0, f_1, f_2),\) where $f_i: K_i \to G_i$, for $i = 0, 1, 2$. We will use a notation where each element of the lattice will have a group element attached in order to represent the functions $f_i$.

The Hamiltonian has the form
\begin{equation}\label{Ham:12-gauge}
H=-\sum_{v \in K_0} A_v -\sum_{l \in K_1} A_l - \sum_{l \in K_1} B_l- \sum_{p \in K_2} B_p - \sum_{c \in K_3}B_c,
\end{equation}

and is composed by a sum of commuting projectors labelled by vertices, links, plaquettes and cubes.
The vertex operator $A_v$ is a sum of gauge transformations. In other words, 

\begin{align}\label{eq:12Av}
A_v=\dfrac{1}{|G_1|}\sum_{g \in G_1} A_v^{g}\,,\,\, \text{where:}\qquad
A_v^g \;\vcenter{\hbox{\includegraphics[scale=0.75]{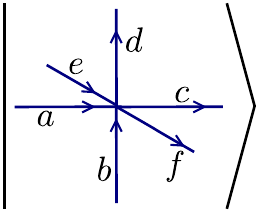}}}\ = \ \vcenter{\hbox{\includegraphics[scale=0.75]{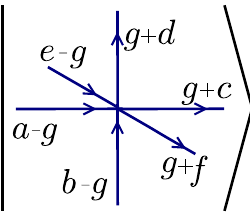}}}.
\end{align}

Higher gauge transformations appear in the definition of $A_l$: \[A_l = \frac{1}{|G_2|}\sum_{\beta \in G_1}A_l^{\beta}.\] The elementary transformation $A_l^{\beta}$ acts on the adjacent plaquettes with $\beta \in G_2$ and on the link with $(\partial_2^G \beta$), as shown in figure \ref{fig:12Al3G}.

\begin{figure}[h!]
	\centering
	\includegraphics[scale=0.8]{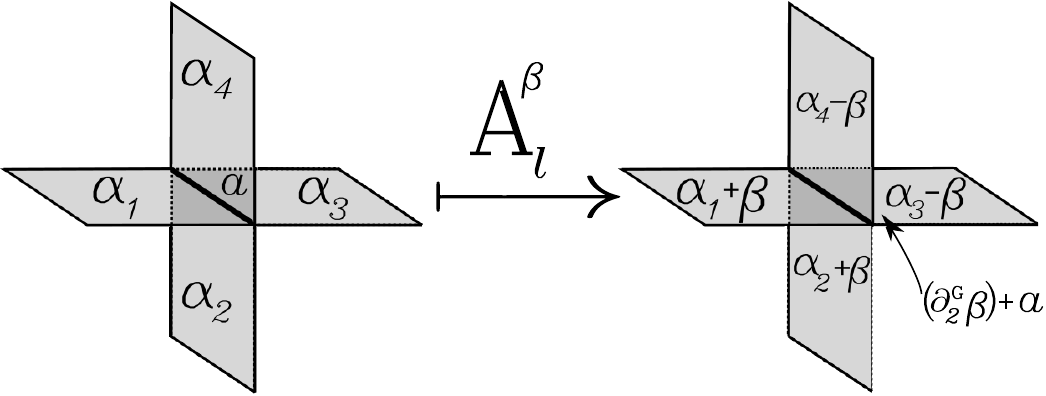}
	\caption{The action of $A_l^\beta$ is shown on an initial arbitrary basis state, involving the link \(l \in K_1\) and its four adjacent plaquettes.}
	\label{fig:12Al3G}
\end{figure}
 The 0-holonomy operator, $B_l$, compares the gauge fields of 
adjacent vertices with the map $\partial_1^G$ applied to the link degree of 
freedom, namely
\begin{equation}
\label{eq:Bl1G}
B_l \;  \vcenter{\hbox{\includegraphics[scale=0.75]{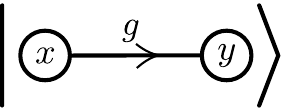}}} = 
\delta (x -y, \partial_1^Gg)\ 
\vcenter{\hbox{\includegraphics[scale=0.75]{Bl01ket.eps}}}.
\end{equation}
The 1-holonomy is not exactly the same as in the Quantum Double models: the operator \(B_p\) is defined by
 \begin{align}
 \label{eq:Bp2G}
 B_p \quad \vcenter{\hbox{\includegraphics[scale=0.75]{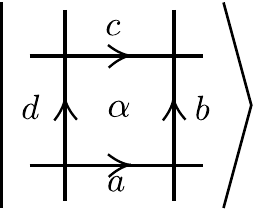}}}\ &= 
 \delta\left(a+b-c-d,\partial_2^G\alpha\right)\ \vcenter{\hbox{\includegraphics[scale=0.75]{Bp2Gket1.eps}}}.
\end{align}
Note that this operator gives eigenvalue \(1\) when the holonomy of the plaquette is equal to \(\partial_2^G\alpha\). To distinguish it from the ordinary holonomy, this operator is usually called fake holonomy. Here we will use the name $1$-holonomy instead. 

\par To simplify our notation, we will denote $\partial_2^G\alpha$ simply by $\partial \alpha $ whenever there is no danger of ambiguity. The $2$-holonomy operator $B_c$ measures the sum of the degrees of freedom living on the plaquettes that compose the boundary of the cube $c$, with eigenvalue $1$ when the result of the sum equals zero and eigenvalue zero otherwise, as defined in equation (\ref{eq:12Bc}):

\begin{align}
\label{eq:12Bc}
B_c \quad \vcenter{\hbox{\includegraphics[scale=0.55]{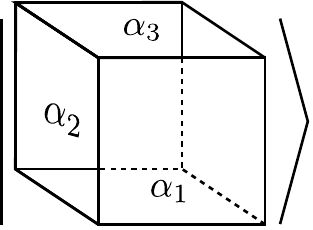}}}\ &= 
\delta(\sum_{j = 1}^6(-1)^{o_j}\alpha_j,0)\ \vcenter{\hbox{\includegraphics[scale=0.55]{12bc3G.eps}}},
\end{align} 
where \(o^j=\{0,1\}\) takes into account the relative orientation of the faces with respect to the cube, similarly to the $2$-holonomy operator. A discussion about the topological properties of this model, such as ground state degeneracy ($GSD$), will be postponed until we present the general formalism.

The data needed to cast this model into the formalism of 
Section \ref{sec:higher} is as follows. One needs to specify a pair of chain complexes. The first one represents the discretization of the manifold and is given by
\begin{equation} 
 0 \hookrightarrow C_3 \xrightarrow{\partial_3} C_2 \xrightarrow{\partial_2} C_{1} \xrightarrow{\partial_1} 
                C_0 \rightarrow 0,
\label{eq.chain}
\end{equation}
where $C_i$ is generated by $K_i$ and $\partial_i$ is the boundary map, for $i = 1,2,3$.
The algebraic data is also a chain complex of the form
\begin{equation} 
0 \hookrightarrow G_2 \xrightarrow{\partial_2^G} G_{1} \xrightarrow{\partial_1^G} 
                G_0 \rightarrow 0.
\end{equation}

\begin{exmp}{$G_0=G_1=\Z_2,\,G_2=\Z_4$}\label{ex:123-Z2Z4}

The groups that label the degrees of freedom are chosen to be $G_0=G_1= \Z_2 = \{0,1\}$ and $G_2=\Z_4=\{0,1,2, 3\}$. Moreover, the homomorphism that relates $1,2$-gauge fields is chosen to be such that $\partial_2^G(1 )=1$. As in the previous example, we use a graphical description for the basis states as follows:

\begin{itemize}
	\item a dotted line through a plaquette when it holds a $\ket{2}_p$ state, and an oriented dashed line for the $\ket{1,3}_p$ states, 
	\item for the links we picture a gray surface orthogonal to it for $\ket{1}_l$ state, no surface for the $\ket{0}_l$ state,
	\item a gray volume for vertices in  $\ket{1}_v$ state.
\end{itemize}
The flat configurations of the theory are those that are invariant under the actions of \(B_l\), \(B_p\) and \(B_c\) for all link \(l\), plaquettes \(p\) and cubes \(c\) of the lattice. Such configurations consist of dotted or dashed loops, conditions that are enforced by \(B_c\). The 1-holonomy operator \(B_p\) implies that every dashed line encloses a gray surface, which is the boundary of a gray volume. If $\Sigma$ has the topology of a $3$-dimensional ball, there is only one ground state given by
\begin{align}
\ket{GS}=&\prod_{v}A_v\prod_{l}A_l\bigotimes_{v}\ket{1}_v\bigotimes_{l}\ket{1}_l\bigotimes_{p}\ket{1}_p. \end{align}
For manifolds with exotic topologies, the number of ground states depends on the number of non-equivalent non-contractible closed surfaces one can drawn over $\Sigma$. The topological nature of this model will become clearer when we introduce the general formalism in \ref{sec:higher}. 
\end{exmp}

Now that we are familiar with some examples of Abelian higher gauge theories, we are ready to describe them all at once using a more general mathematical structure. The more general framework that is going to be exhibited in the next section also allows to compute the entanglement entropy, in the more general case, as we will show in section \ref{sec:entent}.

\section{\label{sec:higher}Review of Abelian Higher Gauge Theories}

The model presented on section \ref{sec:ExampleAHGT} is an example of what we  
call \emph{Abelian higher gauge theories}
introduced in \cite{higher}. 
It is useful to describe them using a formalism borrowed from Homological 
Algebra since it  allows us to handle a large class of models of 
arbitrary dimensions. 
In this section we recall from \cite{higher} only the basic notation and 
results needed to calculate the entanglement entropy for any 
such model. We refer to \cite{higher} for further details.
We have indicated in the last section that the models are parameterized by two chain complexes. 
The first one is geometrical in nature and accounts for the structure of the 
lattice. As for the second, it is a chain complex of finite Abelian groups encoding the 
higher gauge group of the model. There will be one model (Hilbert space 
and Hamiltonian) for any such choice of chain complexes. The choices corresponding to 
the example was given in section \ref{sec:ExampleAHGT}.

A simplicial decomposition is a natural choice for lattices of any dimension. Although 
the formalism can accommodate for any finite cell decomposition we will assume 
that the lattice $K$ is made of simplices. In other words
\[
K=K_0\cup K_1\cup \cdots\cup K_d,
\]
where $K_n$ is the (finite) set of $n$-dimensional simplices. 
We would like to point 
out that there are no further assumptions on $K$, which makes the formalism very flexible.
For instance, $K$ may have a boundary and may not have a uniform dimension

It is a standard procedure \cite{Hatcher} to associate to $K$ a
chain complex
\begin{equation} \label{eq:Cchain}
 0 \to C_d \xrightarrow{\partial^C_d} C_{d-1} \xrightarrow{\partial^C_{d-1}} \cdots  \xrightarrow{\partial^C_2} C_1 \xrightarrow{\partial^C_1} C_0 \to 0,
\end{equation}
which we will denote by $(C(K),\partial^C)$, where $0$ represents the trivial group.
We recall that $C_n$ is the Abelian group freely generated by $K_n$.
In other words,
if we write the group operation as an addition operation, $c\in C_n$ is given by a formal linear combination
\begin{equation}
c=\sum_{x\in K_n} n(x)x
\label{cn-chain}
\end{equation}
with $n(x)\in {\mathbb{Z}}$. The homomorphisms $\partial^C_n:C_n 
\rightarrow C_{n-1}$ are the usual boundary maps. 

To describe the higher gauge groups that label the degrees of freedom in the simplicial complex, we introduce a chain complex \((G, \partial^{G})\) of
finite Abelian groups given by
\begin{equation}\label{eq:Gchain}
0 \rightarrow G_d \xrightarrow{\partial^{G}_d}G_{d-1} \xrightarrow{\partial^{G}_{d-1}}\cdots \xrightarrow{\partial^{G}_2} G_1 \xrightarrow{\partial^{G}_1} G_0 \rightarrow 0,
\end{equation}
where  $0$ denotes the trivial group and $\partial^{G}_n:G_n \rightarrow G_{n-1}$ 
are group homomorphisms such that $\partial^{G}_{p}\circ\partial^{G}_{p+1}=0$, for any $0\le p\le d$. Note that the $d = 2$ case, i.e., the chain
\begin{equation}\label{eq:cross}
	0 \to G_2 \xrightarrow{\partial^G_2} G_1 \to 0 
\end{equation}
can be recast into the language of \textit{strict 2-groups} or equivalently, \textit{crossed modules}, which were used in \cite{Bullivant16} to construct models of topological phases in $3+1$ dimensions based on higher gauge theories. A \textit{crossed module} is a quadruple $\text{G} = (G_1, G_2, \partial, \triangleright)$, where $G_1$ and $G_2$ are groups, $\partial:G_2\to G_1$ is a group homomorphism and $\triangleright: G_1 \times G_2 \to G_2$ is an action of $G_1$ on $G_2$, which satisfies the following conditions:
\begin{align}
	\partial(g \triangleright \alpha) = g(\partial\alpha)g^{-1}, \\
	(\partial\beta)\triangleright\alpha = \beta\alpha\beta^{-1}.
\end{align}  
Since we consider only Abelian groups, the chain (\ref{eq:cross}) defines a crossed module with trivial action. In fact, the general case where the chain complex $(G,\partial^G)$ is composed of $d+1$ groups can be reformulated in the language of (Abelian) strict $(d+1)$-groups.  

Now, we define a gauge configuration $f$ to be an assignment of a group element $g\in G_n$
for each element $x\in K_n$. In other words, a gauge configuration is a sequence \(f=\{f_n\}_{ n=0}^d\) of functions such that
\begin{align}
\label{eq:fKn-G}
f_n: K_n &\rightarrow G_n, \\
     x & \mapsto f_n(x).
\end{align}
Strictly speaking, we should call $f$ a higher-gauge 
configuration. Only in the case when all groups except $G_1$ are trivial
$f$ is a proper gauge configuration, as can be seen from the examples of last section. 
For simplicity, we will keep using "gauge configuration" to mean a generic $f$. 

Because $C_n$ is freely generated by $K_n$, each map $f_n$ in (\ref{eq:fKn-G}) defines a unique group homomorphism $f_n:C_n\rightarrow G_n$, which is given by the extension of $f_n$ by linearity, i.e., if $c\in C_n$ as in (\ref{cn-chain}), then 
\[
f_n(c)=\sum_{x\in K_n} n(x)f_n(x).
\]
We use the same name $f_n$ to denote a gauge configuration as in \ref{eq:fKn-G} and a homomorphism $f_n:C_n\to G_n$ since there is a one to one correspondence between them.
The set $\Hom (C_n,G_n)$ of homomorphisms is also an Abelian group if we set
\[
(f_n+\tilde{f_n})(x)=f_n(x)+\tilde{f_n}(x), ~f_n,\tilde{f_n}\in \Hom (C_n,G_n).
\]
It is useful to collect all such Abelian groups in a single direct sum.
This simple fact allows us to view a gauge configuration $f$ as an element of the direct sum
\begin{equation}
\hom^0 := \bigoplus_{n=0}^d \Hom (C_n,G_n)
\label{hom0}
\end{equation}
of Abelian groups. Thus, a gauge configuration can be represented by a collection of maps between chain complexes, as
depicted  by the diagram in Fig. \ref{fig:classconf}. We would like to point out that figure \ref{fig:classconf} is not a commuting diagram.
When this happens,  $f\in \hom^0$ is called a chain map and, as we will see, the corresponding gauge configuration is gauge equivalent to the trivial one.
\begin{figure}[h!]
\begin{center}
\begin{adjustbox}{max size={.5\textwidth}{.5\textheight}}
\begin{tikzpicture}
  \matrix (m) [matrix of math nodes,row sep=2em,column sep=3em,minimum width=2em]
  {   \cdots & C_{n+1} & C_{n} & C_{n-1} & \cdots  \\
      \cdots  & G_{n+1} & G_{n} & G_{n-1} & \cdots \\};
  \path[-stealth]
    (m-1-1) edge node [above] {$\partial^C_{n+2} $} (m-1-2)
    (m-1-2) edge node [above] {$\partial^C_{n+1} $} (m-1-3)
            edge node [left] {$f_{n+1}$} (m-2-2)
    (m-1-3) edge node [above] {$\partial^C_{n}$} (m-1-4)
    		edge node [left] {$f_{n}$} (m-2-3)
    (m-1-4) edge node [above] {$\partial^C_{n-1}$} (m-1-5)       
    (m-1-4) edge node [right] {$f_{n-1}$} (m-2-4)
    (m-2-1) edge node [below] {$\partial^G_{n+2} $} (m-2-2)
    (m-2-2) edge node [below] {$\partial^G_{n+1} $} (m-2-3)
    (m-2-3) edge node [below] {$\partial^G_{n} $} (m-2-4)
    (m-2-4) edge node [below] {$\partial^G_{n-1} $} (m-2-5);
\end{tikzpicture}
\end{adjustbox}
\end{center}
\caption{\label{fig:classconf} A configuration $f \in \text{hom} (C,G)^0$, consisting on a collection of homomorphisms \(\{f_n\}\).}
\end{figure}

We are now in position to define the model by providing a Hilbert space $\mathcal{H}$ of states $\ket{\psi}$
and a Hamiltonian operator $H$ acting on $\mathcal{H}$. The set $ \hom^0$ is finite. Let $\H$ be the complex vector space generated by the orthonormal basis $\{\ket{f}\}_{f\in \hom^0}$. In other words, a state $\psi \in \H$ is given by a linear combination
\begin{align}\label{eq:gen-state}
\ket{\psi}=\sum_{f\in \hom^0} \psi(f)\ket{f},
\end{align}
where $\psi(f)\in \C$, and the internal product is $\braket{f|g} = \delta(f,g)$. We can explicit the orthonormal basis $\{\ket{f}\}$  as a tensor product over the local degrees of freedom living on the simplexes labeled by the groups.
\begin{align}\label{eq:ortho-basis}
\ket{f}=\bigotimes_{v\in K_0} \ket{f(v)}\bigotimes_{l\in K_1} \ket{f(l)}\bigotimes_{p\in K_2} \ket{f(p)}\bigotimes \cdots  \bigotimes_{s\in K_d} \ket{f(s)},
\end{align}
where the tensor product is made over all the 0-simplexes (vertices $v$), 1-simplexes (links $l$), 2-simplexes (plaquette $p$),... through to the $d$-simplexes associated to a group of the chain $(G, \partial^G)$.
In order to define the Hamiltonian, we will need to introduce more groups other than
$\hom^0$ given by (\ref{hom0}). Let us consider the groups $\hom ^{p}$ defined by
\begin{equation}
\label{eq:homp}
\hom^p := \bigoplus_{n=0}^{d} \Hom(C_n, G_{n-p}).
\end{equation}
An element $g\in \hom^p$ is a sequence $\{g_n\}_{n=0}^d$ of homomorphisms  
\mbox{$g_n:C_n\rightarrow G_{n-p}$}. The example of \(\hom^{1}\) is shown in Fig. \ref{fig:ChainMaps}.  
\begin{figure}[h!]
\begin{center}
\begin{adjustbox}{max size={.7\textwidth}{.7\textheight}}
\begin{tikzpicture}
  \matrix (m) [matrix of math nodes,row sep=2em,column sep=3em,minimum width=2em]
  {   \cdots & C_{n+1} & C_{n} & C_{n-1} & \cdots  \\
      \cdots  & G_{n+1} & G_{n} & G_{n-1} & \cdots \\};
  \path[-stealth]
    (m-1-1) edge node [above] {$\partial^C_{n+2} $} (m-1-2)
    (m-1-2) edge node [above] {$\partial^C_{n+1} $} (m-1-3)
    (m-1-3) edge node [above] {$\partial^C_{n}$} (m-1-4)
    (m-1-4) edge node [above] {$\partial^C_{n-1}$} (m-1-5)  
    (m-2-1) edge node [below] {$\partial^G_{n+2} $} (m-2-2)
    (m-2-2) edge node [below] {$\partial^G_{n+1} $} (m-2-3)
    (m-2-3) edge node [below] {$\partial^G_{n} $} (m-2-4)
    (m-2-4) edge node [below] {$\partial^G_{n-1} $} (m-2-5);
    \path[-stealth]
     (m-1-1) edge node [left] {$ g_{n+2} $} (m-2-2)
    (m-1-2) edge node [left] {$g_{n-1}$} (m-2-3)
     (m-1-3) edge node [left] {$ g_n $} (m-2-4)
     (m-1-4) edge node [left] {$g_{n-1}$} (m-2-5);
\end{tikzpicture}
\end{adjustbox}
\end{center}
\caption{\label{fig:ChainMaps} An element 
$g \in \hom (C,G)^1$ as a sequence  $\{g_n\}_{n=0}^d$ of skewed maps}
\end{figure}

An important observation is that the sequence of groups $\hom^p$ can be made into a  co-chain
complex. This is achieved  by considering maps
\(\delta^p : \hom^p \rightarrow \hom^{p+1}\), defined by:
\begin{equation}\label{eq:cob_op}
(\delta^p h)_n=h_{n-1}\circ \partial^C_n - (-1)^p \partial^G_{n-p}\circ h_n ,
\end{equation}
with \(h \in \hom^p\). In fact, it is straightforward to verify that \(\delta^{p+1} \circ \delta^p = 0\), which turns the sequence  
\begin{align}\label{eq:homcoChain}
\cdots \xrightarrow{\quad\delta^{-2}\;} \text{hom}(C,G)^{-1} \xrightarrow{\quad\delta^{-1}\;} \text{hom}(C,G)^{0} \xrightarrow{\quad\delta^0\;}\text{hom}(C,G)^{1} \xrightarrow{\quad\delta^{1}\;} \cdots 
\end{align}
into a co-chain complex.
The expression above shows only the part of the sequence that is relevant for the present 
application, please refer to \cite{higher} for a more detailed account. Associated to this co-chain complex, there are the so-called \textit{Brown cohomology groups} \cite{Brown}
\begin{align}\label{eq:BrownCohomology}
	\text{H}^p(C,G) = \text{ker}(\delta^{p})/\text{Im}(\delta^{p-1}), 
\end{align}
which, as shown in \cite{Brown}, are isomorphic to the direct product of the cohomology groups of the chain complex $(C(K),\partial^C)$ with coefficients in the homology groups of $(G,\partial^G)$, i.e.,
\begin{align}\label{eq:BrownIsomorphism}
	\text{H}^p(C,G) \cong \bigoplus_{n=0}^d \text{H}^n(C, \text{H}_{n-p}(G)).
\end{align} 
From (\ref{eq:BrownIsomorphism}), it is clear that the cohomology groups defined in equation (\ref{eq:BrownCohomology}) are topological invariants of the manifold described by the chain complex $(C(K),\partial^C)$. The Brown cohomology plays a major role in our formalism, as we will see in the following sections. We refer the reader to \cite{higher, Brown} and references therein for a more detailed account on Brown cohomology and on the isomorphism (\ref{eq:BrownIsomorphism}). 

To complete the description, we need to define a chain complex that is the dual of 
(\ref{eq:homcoChain}). This will be done by dualizing the groups $G_n$ as follows.
Let $\Hom(G_n,U(1))$ be the set of homomorphisms
$a:G_n\rightarrow U(1)$, for each $0\le n \le d$. Since each $G_n$ is Abelian, this is nothing but the set of 
irreducible unitary representations of $G_n$, denoted by $\widehat{G}_n$. Let us give $\widehat{G}_n$ a structure of an Abelian group. 
Let $a,b \in \widehat{G}_n$ and 
$g\in G_n$. Let us write the group operation in $\widehat{G}_n$ 
as $a+b$ and the inverse of $a$ as $-a$.
The group is defined by setting $(a+b)(g)=a(g)b(g)$ and $(-a)(g)=(a(g))^{-1}$. 
In order to dualize (\ref{eq:homcoChain}) we first define the dual $\hom_p$ of (\ref{eq:homp}) as
\begin{equation}
\hom_p := \bigoplus_{n=0}^d \Hom (C_n,\widehat{G}_{n-p}).
\end{equation}

As before, an element $m\in\hom_p$ is a sequence $\{m_n\}_{n=0}^d$ with
$m_n\in \Hom(C_n,\widehat{G}_{n-p})$. Each $m_n$ is completely defined by its values 
on the generators $x\in K_n$. This allows us to introduce a pairing
\begin{align}
\langle\cdot,\cdot\rangle:\hom_p\times \hom^p & \rightarrow U(1) \nonumber \\ 
     (m,f) & \mapsto \langle m,f\rangle
\end{align}
given by
\begin{align}
\langle m,f\rangle=\prod_{n=0}^d \, \prod_{x \in K_n} m_n(x)(f_n(x)).
\end{align}
Let us define  a boundary map $\delta_p:\hom_p \rightarrow \hom_{p-1}$ 
given by
\begin{equation}
\label{eq:deltaHomChain}
\langle \delta_pm,f\rangle=\langle m,\delta^{p-1}f\rangle ,
\end{equation}
where $m\in \hom_p$ and $f\in \hom^{p-1}$. Clearly, $\delta_p \circ \delta_{p+1}=0$ and thus the
chain complex dual to (\ref{eq:homcoChain}) that we will need is given by 
\begin{align}\label{eq:homChain}
\cdots \xleftarrow{\quad\delta_{-1}\;}\text{hom}(C,G)_{-1} \xleftarrow{\quad\delta_{0}\;} \text{hom}(C,G)_{0} \xleftarrow{\quad\delta_1\;}\text{hom}(C,G)_{1} \xleftarrow{\quad\delta_{2}\;} \cdots.
\end{align}

\subsection{Operators and Hamiltonian}

The Hamiltonian we presented as an example in section \ref{sec:ExampleAHGT}
have a similar structure. It is a sum of operators that can be divided in
two types. There are higher gauge transformations and diagonal operators 
measuring higher holonomies. 
In the general formalism, these two sets of operators in $\mathcal{H}$ come 
from the co-chain complex (\ref{eq:homcoChain}) and the chain complex 
(\ref{eq:homChain}), respectively.
The first set of operators is parametrized by $\hom^{-1}$ whereas the second one 
by $\hom_{1}$.

For $t\in \hom^{-1}$ and $m\in \hom_{1}$ 
we define:
\begin{align}\label{eq:At}
A_t \ket{f} &: = \ket{f + \delta^{-1} t},\\
B_m \ket{f} &: = \langle m,\delta^0 f \rangle \ket{f}.
\label{eq:Bm}
\end{align}

The interpretation of (\ref{eq:At}) and (\ref{eq:Bm}) can be 
derived from the special case when the chain complex 
(\ref{eq:Gchain}) is made of trivial groups except for $G_1$.
The resulting model has the familiar form of a gauge theory on the lattice.
In this case, it follows that a configuration $f\in \hom^0$ assigns one group element of $G_1$ for each link of the lattice, as expected in a ordinary
gauge theory. Furthermore, $t\in \hom^{-1}$  gives a group element $g(v)\in G_1$ for each
vertex $v$ of the lattice. One can verify that $A_t$ performs on each vertex 
$v$ an ordinary gauge transformation with parameter $g(v)$. 
As for the general case, $f$ and $A_t$ define what we mean by higher 
gauge configurations and higher gauge transformations.
As we pointed out before, we will keep calling them gauge configurations and
gauge transformations. Going back to the special case, we need to look at the 
eigenvalue $\langle m,\delta^0 f \rangle$ of $B_m$ to see what is it 
measuring. It follows from the definition that 
$\delta^0f\in \hom^1$ and $m\in \hom_1$, it assigns to each face $p$ of the lattice 
its holonomy $h_p$ and a unitary representation $\chi_p\in U(1)$ 
respectively.
Operators $B_m$ are therefore measuring the number $\prod_p \chi_p(h_p)$ 
depending on all holonomies of the lattice. 

The first thing to be noticed is that both $A_t$ and $B_m$ are not localized as they
act on the entire lattice. For the definition of the Hamiltonian, however, we need 
to define local projectors. This is easily achieved by taking $t$ and $m$ with a 
local support in the lattice $K$ and averaging over the groups.  

\begin{defn}[Localized maps]\label{def:locmaps}
Let $x\in K_n$, $g \in G_{n+1}$ and $r \in \hat{G}_{n-1}$. We define the local maps
$\hat{e}[n,x,r]\in \text{hom}(C,G)_{1}$ and $e[n,x,g] \in \text{hom}(C,G)^{-1}$ by
\begin{align}
\label{eq:e[x,g]}
e[n,x,g] (y)  :=& \begin{cases} g, \text{ if } y=x \\ 0, \text{ otherwise } \end{cases}    \\ \quad \hat{e}[n,x,r](f) :=& r(f_n(x)),
\label{eq:e[x,r]} 
\end{align}
where $y \in K$, and $f\in \text{hom}(C,G)^p$.
\end{defn}

\begin{defn}[Local projector operators]\label{def:AxBx}
Let $x\in K_n,g \in G_{n+1}$, $r \in \hat{G}_{n-1}$. We define \textit{local gauge projector} $A_{n,x}$ and \textit{local holonomy projector} $B_{n,x}$ as:
\begin{align}
A_{n,x} &=\frac{1}{|G_{n+1}|}\sum_{g\in G_{n+1}}  A_{e[n,x,g]} ,\; \quad  \\  
B_{n,x} &=\frac{1}{|G_{n-1}|}\sum_{r\in \hat{G}_{n-1}}  B_{\hat{e}[n,x,r]} .\;
\end{align}
\end{defn}

The Hamiltonian operator $H$ is defined as
\begin{align}\label{eq:Hamiltonian}
H =-\sum_{n=0}^{d} \sum_{x\in K_n} A_{n,x} -\sum_{n=0}^d \sum_{x\in K_n} B_{n,x}.
\end{align}

It is straightforward to show that $A_{n,x}$ and $B_{n,x}$ are commuting projectors. Furthermore, in the special case where the chain complex 
$(G,\partial^G)$ has only $G_1$ different from the trivial group, we recover
the quantum double model with group $G_1$. Also, by choosing  $(G,\partial^G)$
we can reproduce the example we have discussed in the last section.

This Hamiltonian is actually frustration free since there is at least one
state that gives eigenvalue $1$ for all $A_{n,x}$ and $B_{n,x}$. 
Let $\ket{0}$ denotes the state
labeled by the trivial element of the group $\hom^0$. It corresponds to a configuration that maps all elements of $K_n$ to $0\in G_n$ for all $n$. 
Let us define
\begin{align*}
\ket{0}_G = \prod_{n=0}^d \prod_{x \in K_n}\, A_{n,x} \ket{0}.
\end{align*} 
One can show that $\ket{0}_G$ is non zero and
\begin{align}
A_{n,x}\ket{0}_G & =\ket{0}_G,\\
B_{n,x}\ket{0}_G & =\ket{0}_G,
\end{align}
for all $0\leq n \leq d$ and $x \in K_n$. Therefore, a state $\ket{\psi}$ is in the 
ground state ${\cal H}_0$ if and only if 
\begin{align}\label{eq:GSeq} 
A_{n,x}\ket{\psi} & =\ket{\psi},\\
B_{n,x}\ket{\psi} & =\ket{\psi},
\end{align}
for all $0\leq n \leq d$ and $x \in K_n$.

It is useful to characterize ${\cal H}_0$ in another way. Let us consider
the following operators:
\begin{enumerate}
\item projector $\mathcal{A}_0$ given by
\begin{align}\label{eq:projA0}
\mathcal{A}_0=\dfrac{1}{|\text{hom}(C,G)^{-1}|}\;\sum_{t \in \text{hom}(C,G)^{-1}} A_t,
\end{align}
that maps any state $\ket{f} \in \mathcal{H}$ into a normalized sum of gauge equivalent states;
 \item projector $\mathcal{B}_0$ given by
\begin{align}\label{eq:projB0}
\mathcal{B}_0=
\dfrac{1}{|\text{hom}(C,G)_{1}|}\;\sum_{m\in \text{hom}(C,G)_{1}} B_m,
\end{align}
 which gives eigenvalue \(1\) for a state 
 $\ket{f} \in \mathcal{H}$, only if satisfies $f \in \text{ker} (\delta^0)$; in other words, it projects onto the \emph{flat holonomy} sector of \(\H\).
\end{enumerate}

As stated in \cite{higher}, the projector $\Pi_0$ on the ground state subspace 
${\cal H}_0$ can be written as
\begin{align}
\label{def:P0}
\Pi_0 := \mathcal{A}_{0} \mathcal{B}_{0} \;.
\end{align}
Furthermore, the dimension of ${\cal H}_0$ is determined by the zeroth Brown cohomology group $\text{H}^0(C,G)$ of the cochain complex in Eq.(\ref{eq:homcoChain}), which explicitly shows how the ground state degeneracy depends on the topology of the underlying manifold. This result can be stated more precisely as follows:

\begin{thm}[Dimension of the ground state subspace]\label{thm:main}
The dimension $\text{GSD} = \text{dim}(\mathcal{H}_0)$ of the ground state subspace $\mathcal{H}_0$ is given by the number of \emph{flat} states \(|\text{ker}(\delta^0)|\), modulo the \emph{gauge} equivalence \(|\text{Im}(\delta^{-1})|\), that is
\begin{align}\label{eq:GSD}
\text{GSD} = \frac{ |\text{ker} (\delta^0)|}{|\text{Im}(\delta^{-1})|} = |\text{H}^0(C,G)| = \prod_{n=0}^d|\text{H}^n(C, \text{H}_{n}(G))|.  
\end{align}
\end{thm}
\begin{proof}
The proof of this theorem can be found in \cite{higher} as its main result.
\end{proof}

\subsection{Example}
To see how this general framework works, let us review the example given in section \ref{sec:ExampleAHGT}, this time built from the formalism presented in \S \ref{sec:higher}.

\subsubsection{0,1,2-Gauge}\label{ex:12-chain}
The 0,1,2-gauge model presented in section \ref{ex:12} comes from the chain complex \((C,\partial^G)\) given in equation \eqref{eq.chain}. We show the relevant maps in figure \ref{fig:ChainMaps-1,2}, in particular:
\begin{itemize} 
\item Classical gauge configurations now consider degrees of freedom on plaquettes given by maps $f_2\in \Hom (C_2,G_2)$, in addition to the link configurations defined by $f_1\in \Hom (C_1,G_1)$, and the configurations for the vertices $f_0\in \Hom (C_0,G_0)$. In figure \ref{ex:12-chain} they are represented as straight lines.
\item The generalized notion of gauge transformations include 1-gauge transformations given by maps $t_0\in \Hom (C_0,G_1)$ and 2-gauge transformation, from every link to its neighbor plaquettes $t_1\in \text{Hom}(C_1,G_2)$. In figure \ref{ex:12-chain} they are represented as skewed dotted lines.
\item We have as holonomy maps: the 1-holonomy with the functions $m_1\in\Hom (C_1,\hat{G}_0)$, the 1-holonomy which is measured by maps $m_2\in\Hom (C_2,\hat{G}_1)$ and the $2$-holonomy measured by \(m_3\in \Hom (C_3,\hat{G}_2)\). They are nt represented in the figure, but they would be skewed in the opposite direction of the gauge transformations.
\end{itemize}
So the Hamiltonian in equation \eqref{Ham:12-gauge} is obtained by the decomposition for the maps:
\begin{align}
H =- \sum_{x\in K_0} \frac{1}{|G_1|}\sum_{g\in G_1}A_{e[0,x,g]}  - \sum_{x\in K_1} \frac{1}{|G_2|}\sum_{g\in G_2}A_{e[1,x,g]} \; + \nonumber \\ - \; \sum_{y\in K_1} \frac{1}{|G_0|}\sum_{r\in \hat{G}_0} B_{\hat{e}[1,y,r]}  - \; \sum_{y\in K_2} \frac{1}{|G_1|}\sum_{r\in \hat{G}_1} B_{\hat{e}[2,y,r]} - \sum_{y\in K_3} \frac{1}{|G_2|}\sum_{r\in \hat{G}_2} B_{\hat{e}[3,y,r]}.
\end{align}
\begin{figure}[h!]
	\begin{center}
		\begin{adjustbox}{max size={.7\textwidth}{.7\textheight}}
			\begin{tikzpicture}
			\matrix (m) [matrix of math nodes,row sep=2em,column sep=3em,minimum width=2em]
			{  0 & C_{3} & C_{2} & C_{1} & C_{0} & 0 \\
				& 0 & G_{2} & G_{1} & G_0 &0 \\};
			\path[-stealth]
			(m-1-1) edge node [above] {$$} (m-1-2)
			(m-1-2) edge node [above] {$\partial^C_3 $} (m-1-3)
			(m-1-3) edge node [above] {$\partial^C_{2} $} (m-1-4)
			edge node [left] {$f_2$} (m-2-3)
			(m-1-4) edge node [above] {$\partial^C_{1}$} (m-1-5)
			edge node [right] {$f_1$} (m-2-4)
			(m-1-5) edge node [above] {$$} (m-1-6)          		
			(m-2-2) edge node [below] {$ $} (m-2-3)
			(m-2-3) edge node [below] {$\partial^G_{2} $} (m-2-4)
			(m-2-4) edge node [below] {$ $} (m-2-5)
			(m-2-5) edge node [below] {$ $} (m-2-6);
			\path[dotted,->]    
			(m-1-2) edge node [left] {$\;m_{3}$} (m-2-3) 
			(m-1-4) edge node [right] {$\;t_{1}$} (m-2-3) 
			(m-1-5) edge node [right] {$\; t_{0}$} (m-2-4);
			\end{tikzpicture}
		\end{adjustbox}
	\end{center}
	\caption{\label{fig:ChainMaps-1,2} Chain complexes for the 0,1,2-gauge model \ref{ex:12-chain}.}
\end{figure}

Now that we have familiarized with the general theory with the help of the example, we are ready to proceed onto the next section, where we present the main results of this paper. The calculation of the entanglement entropy for all Abelian higher gauge theories can be carried out in general. The result obtained relates the entanglement entropy of an Abelian higher gauge theory to the \(GSD\) of a related theory, as we will precisely see. Furthermore we also compute the topological contribution to the entanglement entropy.

\section{\label{sec:entent}Entanglement Entropy in Abelian Higher Gauge Theories}
In this section we calculate the entanglement entropy for the class of models defined in \cite{higher} and reviewed in section \ref{sec:higher}. We begin by defining the bipartition of the  \(\left( C(K), \partial^C \right)\) chain complex into a subcomplex $\left(C(K_A),\partial^C_A\right)$ and its complement. We then observe that an associated higher gauge theory can be defined in the subcomplex $(C(K_A),\partial^C_A)$ which will be useful for both the calculation and the interpretation of the results. As usual, we begin by introducing the density matrix \(\rho\) in terms of the ground state projector of \eqref{def:P0}. The reduced density matrix \(\rho_A = \text{Tr}_B (\rho)\) is then obtained and shown to be best written in terms of the local operators of the higher gauge theory defined in the subcomplex \(\left(C(K_A),\partial^C_A\right)\). The entanglement entropy is the \emph{von Neumann} entropy of the reduced density matrix
\[S_A:= - \text{Tr} (\rho_A \log \rho_A).\]
The result we obtain relates this quantity to a restricted gauge theory in region \(A\). In particular, we show that the entanglement entropy of a higher gauge theory with Hamiltonian as in \eqref{eq:Hamiltonian} is equal to the logarithm of the ground state degeneracy \(GSD_{\tilde{A}}\) of a related higher gauge theory restricted to region \(A\), in other words
\[S_A = \log \left(GSD_{\tilde{A}}\right).\]
We further analyze the result to extract the topological information. As we will explain explain in section \ref{eq:sab_satop}, the entropy $S_A$ has two terms:
\begin{align}
	S_A = S_{\partial A} + S_{\text{Topo}},
\end{align}
where $S_{\partial A}$ scales with the size of the boundary and $S_{\text{Topo}}$ is a constant contribution depending on the topology of $A$ and its boundary $\partial A$. 

\subsection{Bipartition of the Geometrical Chain Complex}\label{sec:bipart}
We recall from Section \ref{sec:higher} that the geometrical content of the model is given by the chain complex \((C(K), \partial^C)\) encoding the lattice. We consider a simplicial chain complex for convenience. In order to calculate the entanglement entropy we first need to define the bipartition of the lattice. The system is divided into two regions $A$ and $B$, where $A$ is the region we have access, as in \cite{Hamma2005,Schmitz2018, Fattal2004}. 

We split the simplicial complex $K = \bigcup_{n=0}^{d}K_n$ into a subcomplex $K_A$ of dimension $d$ by choosing a smaller subset $K_{d,A} \subset K_d$  together with their subcomplexes of smaller dimension. Then \(K_A\) is a subcomplex as the boundary maps are well defined, i.e. their image belong to the subcomplex $K_A$. For each $0\leq n \leq d$ dimension the sets of \(n\)-simplices is divided in the form $K_n = K_{n,A} \cup K_{n,B}$, where $K_{n,A}$ is the set of $n$-simplices in region $A$ and $K_{n,B}$ the set of simplices in region $B$. We do this in such a way that $K_A = \bigcup_{n=0}^{d}K_{n,A}$ is a subcomplex of $K$. Note that in general the complement $K_B$ is not a simplicial complex on its own, as its boundary maps may have image in $K_A$.

Let $C_{n,A}$ be the $n$-chain group generated by the \(n\)-simplices, $x \in K_{n,A}$, of region \(A\).  Let also $\partial^{C}_{n,A}: C_{n,A} \rightarrow C_{n-1,A}$ be the restriction of the boundary map $\partial^C_n$ into the subset $K_{n,A}$. Clearly \(\partial_{n,A}^C \circ \partial_{n+1,A}^C = 0\). This makes \(\left(C(K_A), \partial^{C}_{A}\right)\) into a chain complex
\begin{equation} \label{eq:CAchain}
C_{d,A} \xrightarrow{\partial_{d,A}^C} C_{d-1,A} \xrightarrow{\partial_{d-1,A}^C} \cdots  \xrightarrow{\partial_{2,A}^C} C_{1,A} \xrightarrow{\partial_{1,A}^C} C_{0,A}.
\end{equation}
Let us apply the construction reviewed in Section \ref{sec:higher} to the \(\left(C(K_A), \partial^{C}_{A}\right)\) complex together with the same chain complex of Abelian groups in \eqref{eq:Gchain}, namely
\begin{equation}\label{eq:Gchain2}
0 \rightarrow G_d \xrightarrow{\partial_d^G}G_{d-1} \xrightarrow{\partial_{d-1}^G}\cdots \xrightarrow{\partial_2^G} G_1 \xrightarrow{\partial_1^G} G_0 \rightarrow 0.
\end{equation}

Homomorphisms between the two chain complexes (\ref{eq:CAchain}) and (\ref{eq:Gchain2}) can be constructed giving rise to the groups \[\homA^{p}:=\bigoplus_{n} \text{Hom}(C_{n,A},G_{n-p}).\] Elements of such groups are sequences of morphisms \(f_{n,A}:C_{n,A}\, \rightarrow \, G_{n-p}\) whose support lies on \(K_A\). For example, a gauge configuration on region \(A\) is an assignment of a group element \(g \in G_n\) for each element \(x \in K_{n,A}\). This is, a collection of maps \(f_A=\{f_{n,A}\}\) for \(n=0,1,\dots, d\), where:
\begin{align*}
f_{n,A}:K_{n,A}\, & \rightarrow \, G_{n}, \\
x & \mapsto f_{n,A}(x).
\end{align*}
Therefore, gauge configurations in \(A\) can be viewed as elements of the group  \[\homA^{0}=\bigoplus_{n} \text{Hom}(C_{n,A},G_{n}).\]
The set of vectors \(\{\ket{f}_A\}\), labeled by group elements $f \in \homA^{0}$, form a basis of the Hilbert space $\mathcal{H}_A$. In other words, a state \(\ket{\Psi}_A \in \mathcal{H}_A\) is written as
\begin{align*}
\ket{\Psi}_A = \sum_{f\in\homA^0}\Psi(f)\ket{f}_A.
\end{align*}
Similarly, we have the group $\homA^{-1}$ whose elements serve as parameters for the higher gauge transformations as well as the dual group $\homA_{1}$ that parametrizes the higher gauge holonomy operators. More importantly, a higher gauge theory in region \(A\) can be defined by considering a chain complex similar to the one in \eqref{eq:homcoChain}, that is, 
\begin{align}\label{eq:homcochainA}
\homA^{-1} \xrightarrow{\delta_A^{-1}} \homA^{0}\xrightarrow{\delta_A^{0}} \homA^1,
\end{align}
where the co-boundary map \(\delta_A^{p}: \homA^{p} \rightarrow \homA^{p-1}\) is defined by
\begin{align*}
(\delta_A^{p}f)_n:=f_{n-1,A}\circ \partial_{n,A}^{C} - (-1)^p\partial_{n-p}^{G}\circ f_{n,A}. 
\end{align*}
\subsection{Reduced Density Matrix}\label{sec:redrho}

As usual, we start by introducing the density matrix \(\rho\) of the model with Hamiltonian \eqref{eq:Hamiltonian}, given by
\begin{equation}\label{eq:rho}
\rho := \dfrac{\Pi_0}{\text{tr}\left(\Pi_0\right)}=\dfrac{\Pi_0}{GSD},
\end{equation}
where \(\Pi_0: \mathcal{H} \rightarrow \mathcal{H}_0\) is the ground state projector of eq. \eqref{def:P0} and \(GSD\) stands for the \emph{ground state degeneracy} of eq. \eqref{eq:GSD}. We start from a product state, a linear combination of the ground space states, which are independent states.

From \eqref{def:P0} we know that the ground state projector \(\Pi_0\) can be written in terms of the projectors in \eqref{eq:projA0} and \eqref{eq:projB0} as
\begin{align}\label{eq:P0}
\Pi_0 =& \left(\dfrac{1}{|\hom^{-1}|}\; \sum_{t \in \hom^{-1}}\, A_t\,\right)\left(\dfrac{1}{|\hom_{1}|}\; \sum_{m \in \hom_{1}}\, B_m\,\right).
\end{align}

However, we want to re-parametrize the two sums in the above equation such that they run over non trivial elements only. In other words, we want to factor the redundancies out of the sums, reframing the sum as sum over classes. This can be achieved by looking at the group structure of \(\hom^{-1}\) and \(\hom_1\). Take for instance \(\hom^{-1}\) whose elements parametrize the higher gauge transformations of the theory. The redundancies in the sum over \(t \in \hom^{-1}\) of \eqref{eq:P0} come from elements that act trivially over quantum states (examples of such elements are shown in Section \ref{sec:Examples}). Recall that gauge transformations act on actual states by means of the \(\delta^{-1}\) operator. Thus, we can identify the elements of \(\hom^{-1}\) that act trivially on states: they form a subgroup of \(\hom^{-1}\) called the \emph{kernel} and given by \(\text{ker}(\delta^{-1}) := \{t \in \hom^{-1} \,|\, \delta^{-1}(t)=0\}\), where \(0 \in \hom^0\) is the identity element that labels the trivial gauge configuration. Morever, non-trivial gauge transformations are parametrized by elements of \(\hom^{-1}\) that are not mapped to the identity by \(\delta^{-1}\), they define a subgroup of \(\hom^0\) known as \emph{image} and denoted \(\text{Im}(\delta^{-1})\). Both the kernel and the image of the co-boundary map, \(\delta^{-1}\), are related to each other by the first isomorphism theorem \cite{basicalg} which in this case reads
\begin{align}\label{eq:1-iso}
\dfrac{\hom^{-1}}{\text{ker}(\delta^{-1})} \simeq \text{Im}(\delta^{-1}).
\end{align}
Elements of the quotient group in the above expression are the cosets of \(\ker (\delta^{-1})\) in \(\hom^{-1}\). This is: \[\frac{\hom^{-1}}{\ker (\delta^{-1})} := \{[t]\;|\, t \in \hom^{-1}\},\] where the coset \([t] = \{t + h_i, \, h_i \in \ker(\delta^{-1})\}\) consists on all elements of \(\hom^{-1}\) that differ from \(t\) by an element in \(\ker(\delta^{-1})\). This is precisely what we need to factor the sums in \eqref{eq:P0}. The sum over \(t \in \hom^{-1}\) can be replaced by a sum over the cosets of \(\ker(\delta^{-1})\) in \(\hom^{-1}\) as follows:
\[\sum_{t \in \hom^{-1}}\, A_t=\sum_{[s] \in \frac{\hom^{-1}}{\text{ker}(\delta^{-1})} }|\ker(\delta^{-1})|\,A_s,\] 
where \(s \in [s]\) is an arbitrary representative of the coset.
A similar argument holds for the sum over \(m \in \hom_1\) which allows to factor out the redundancies from the second sum in \eqref{eq:P0}. By doing this, we ensure that the sums run over independent group elements only:
\begin{align}\label{eq:P0-factored}
\Pi_0=& \left(\dfrac{1}{|\text{Im}(\delta^{-1})|}\sum_{[t] \in \frac{\hom^{-1}}{\text{ker}(\delta^{-1})} }A_t\right)\left(\dfrac{1}{|\text{Im}(\delta_{1})|}\sum_{[m] \in \frac{\hom_1}{\text{ker}(\delta_{1})} }B_m\right),
\end{align}
note that  we have used \(|\hom^{-1}| = |\ker(\delta^{-1})||\text{Im}(\delta^{-1})|\) to simplify the normalization factor of the first sum. A similar identity holds for the second sum. This leaves us with the density matrix of \eqref{eq:rho} written as:
\begin{align}\label{rho}
\rho = \dfrac{1}{GSD}\dfrac{1}{|\text{Im}(\delta^{-1})||\text{Im}(\delta_{1})|}\left(\sum_{[t] \in \frac{\hom^{-1}}{\text{ker}(\delta^{-1})}}A_t\right)\left(\sum_{[m] \in \frac{\hom_1}{\text{ker}(\delta_{1})} }B_m\right).
\end{align}

We can now proceed to the calculation of the reduced density matrix. Let us consider the bipartition of the geometric chain complex  \(\left( C(K), \partial^C \right)\) described in section \ref{sec:bipart}. This procedure splits the Hilbert space into two subspaces \(\mathcal{H}=\mathcal{H}_A\otimes \mathcal{H}_B\). Then, we obtain the reduced density matrix by taking the partial trace over region \(B\), this is:
\begin{align}
\rho_A := \text{Tr}_B (\rho).
\end{align}

To evaluate the partial trace we consider a basis \(\{\ket{f_{n,B}}\}\), where \(f_{n,B} \in \homB^0\) is the restriction of the tensor product \ref{eq:ortho-basis} over simplexes only belonging to $B$. Therefore we can write the basis dividing the tensor product:
\begin{align}
\ket{f} = \ket{f_A}\otimes \ket{f_B},
\end{align}
where $\ket{f_A}$ is the tensor product over the elements belonging to the subcomplex $A$, the same for complementary $B$. For simplicity, let us denote its basis as \(\{\ket{b_i}\} \), with $i = 1, 2 \dots, \dim(\mathcal{H}_B)$. The reduced density matrix is now written:
\begin{align}\label{eq:red-pretrace}
\rho_A = \dfrac{1}{GSD}\dfrac{1}{|\text{Im}(\delta^{-1})||\text{Im}(\delta_{1})|}\sum_{i} \bra{b_i}\left(\sum_{[t]} \sum_{[m]}\,A_t B_m\right)\ket{b_i},
\end{align}
where \([t]\in \frac{\hom^{-1}}{\text{ker}(\delta^{-1})}\) and \([m] \in \frac{\hom_1}{\text{ker}(\delta_{1})}\). Both $A_t$ and $B_m$ are traceless operators unless they are equal to the identity operator, as we show in Proposition \ref{prop:traceless} of Appendix \ref{sec:app2}. Because of this, the only terms that survive the partial trace are those for which $A_t$ and $B_m$ act as identity in \(\mathcal{H}_B\). In other words, the only operators that survive the partial trace are those that act exclusively on \(\H_A\). In terms of local operators this means that all holonomy operators labeled by simplexes that belong to region \(A\) will survive the trace. The case of gauge transformations is more subtle since we need to discard gauge transformations labeled by elements that lie at the boundary of \(A\) as well, see Appendix \ref{sec:app2}.

To account for such operators we use the restricted gauge theory defined by the complex in \eqref{eq:homcochainA} with a slight modification. We are interested on further restricting such a theory by allowing gauge transformations that act on the interior of $K_A$ only. In other words, we want to discard the gauge transformations at the boundary \(\partial A\). To this intent, let us give a more precise notion of \emph{interior} of $A$. Let \(K_{n,\tilde{A}}=\{x \in K_{n,A}\, |\, x \cap \partial A = \emptyset \}\) be the set of all \(n\)-simplices that have no intersection with the boundary of \(A\). Then, the \emph{interior} of \(A\) is the set \(\tilde{A} : = \bigcup_{n=0}^{d} K_{n, \tilde{A}}\label{def:intA}\). 

Recall from section \ref{sec:higher} that the higher gauge transformations of a higher gauge theory are parametrized by elements of the group \(\hom^{-1}\). Thus, to account for gauge transformations that act exclusively on the interior of \(A\) we just need to consider the subgroup of  \(\homA^{-1}\)  whose support is contained on \(\tilde{A}\) only. This can be done using the notion of interior of \(A\) as follows:  consider homomorphisms whose support lie on the interior of $A$, namely \(\text{Hom}(C_{n,\tilde{A}},G_{n+1})\). Let then: 
\[\homtA^{-1} := \bigoplus_{n=0}^{d}\text{Hom}(C_{n,\tilde{A}},G_{n+1}),\]
with elements \(f \in \homtA^{p}\) consisting on collections of maps \(f=\{f_n\}\):
\begin{align}
f_n: K_{n,\tilde{A}} &\rightarrow G_{n+p}, \\
     x & \mapsto f_n(x),
\end{align}
where \(x \in K_{n,\tilde{A}}\) and \(f_n(x) \in G_{n+p}\).

It is straightforward to show that \(\homtA^{p} \) is  a subgroup of \(\homA^{p}\). Moreover, we can define the restriction of the co-coundary operator \(\delta^p\) into the interior of $A$. This is, \(\delta_{\tilde{A}}^p:= \delta^p\vert_{\tilde{A}}\), such that the sequence
\begin{align}\label{eq:redhomcochainA}
\homtA^{-1} \xrightarrow{\delta_{\tilde{A}}^{-1}} \homA^{0}\xrightarrow{\delta_A^{0}} \homA^1,
\end{align}
is a co-chain complex, i.e., \(\delta_A^{0}\circ\delta_{\tilde{A}}^{-1}= 0\). This co-chain complex encodes an Abelian higher gauge theory over \(K_A\) whose gauge transformations are restricted to act on the interior of \(A\) only.

We can now return to equation \eqref{eq:red-pretrace} and evaluate the partial trace of the density matrix \(\rho\), which yields 
\begin{align}\label{eq:rhoA1}
\rho_A = \dfrac{1}{GSD}\dfrac{1}{|\text{Im}(\delta^{-1})||\text{Im}(\delta_{1})|}\left(\sum_{p, q }A_{p} B_{q}\right)\text{Tr}_B(\mathbb{1}_B),
\end{align}
where the sums now run over independent internal gauge transformations 
\begin{align}\label{eq:p} 
p \in \frac{\homtA^{-1}}{\text{ker}(\delta_A^{-1})},
\end{align}
and non-trivial holonomy values in $A$
\begin{align}\label{eq:q} 
 q \in \frac{\homA_1}{\text{ker}(\delta_{1}|_A)}.
\end{align} 
Finally, observing that \(\text{Tr}(\mathbb{1}_B)= \dim (\mathcal{H}_B)\) we get for the reduced density matrix:
\begin{align*}
\rho_A = \dfrac{1}{GSD}\dfrac{\dim( \mathcal{H}_B)}{|\text{Im}(\delta^{-1})||\text{Im}(\delta_{1})|}\left(\sum_{p, q }A_{p} B_{q}\right).
\end{align*}

The above expression can be further simplified as follows: Observe that by applying the first isomorphism theorem \cite{basicalg} on the sequence of Eq.(\ref{eq:homcoChain}) it is easy to show that the dimension of the Hilbert space factors into \[\dim (\mathcal{H})=\left|\hom^0\right|=\left|\text{ker}(\delta^0)\right|\left|\text{Im}(\delta^{0})\right|.\] Also, in Appendix \ref{sec:app3} we show that \(\left|\text{Im}(\delta^0)\right|=\left|\text{Im}(\delta_{1})\right|\) this allows us to write:
\[GSD \,|\text{Im}(\delta^{-1})||\text{Im}(\delta_{1})| =  \dim(\mathcal{H})=\dim(\mathcal{H}_A) \dim(\mathcal{H}_B),\]
which in turn yields for the reduced density matrix,
\begin{equation}\label{eq:rhoA2}
\rho_A = \dfrac{1}{\dim(\mathcal{H}_A)}\left(\sum_{p, q }A_{p} B_{q}\right).
\end{equation}

\subsection{Entanglement Entropy}\label{sec:EE}
Having found the reduced density matrix in \eqref{eq:rhoA2} we are able to calculate its Von Neumann entropy, also known as \emph{entanglement entropy}. This calculation will require us to evaluate the logarithm of \(\rho_A\) at some point and this is usually done using a series expansion. In this sense, we will start by calculating the square of \(\rho_A\):
\begin{align*}
\rho_A^2 =\; \dfrac{1}{\dim(\mathcal{H}_A)^2}\left(\sum_{p, q }A_{p} B_{q}\right)\left(\sum_{p^\prime, q^\prime }A_{p^\prime} B_{q^\prime}\right)\;
 =\: \dfrac{|\text{Im}(\delta_A^{0})||\text{Im}(\delta_{\tilde{A}}^{-1})|}{\dim(\mathcal{H}_A)^2}\left(\sum_{p, q }A_{p} B_{q}\right),
\end{align*}
where in the last equality the factors in the numerator come from rearranging the sums over \(p^\prime \in  \frac{\homtA^{-1}}{\text{ker}(\delta_A^{-1})}\) and over \(q^\prime \in \frac{\homA_1}{\text{ker}(\delta_{1}|_A)}\). This leaves for the square of the density matrix:
\begin{align}
\rho_A^2 = \dfrac{|\text{Im}(\delta_A^{0})||\text{Im}(\delta_{\tilde{A}}^{-1})|}{\dim(\mathcal{H}_A)} \rho_A
= \lambda\,  \rho_A.
\end{align}
Now we can calculate the logarithm of \(\rho_A\) by series expansion, which yields: $\log(\rho_A) = \frac{\log(\lambda)}{\lambda}\rho_A.$
Finally the entanglement entropy is:
\begin{align}\label{eq:SAlambda}
S_A = -\text{Tr}(\rho_A \log(\rho_A))
 = -\text{Tr}(\rho_A \log(\lambda))=
 \log(1/\lambda) \text{Tr}(\rho_A)
 = \log(1/\lambda),
\end{align}
 where we have used $\text{Tr}(\rho_A)=1$. Let us look at the \(\lambda\) factor more carefully, since it encodes the essential information about the entanglement entropy of the model. By recalling that \(\dim(\mathcal{H}_A) = \left|\homA^0\right|=\left|\text{ker}(\delta_A^0)\right|\left|\text{Im}(\delta_A^{0})\right| \), we are able to write:
\begin{align}\label{eq:lambda}
\dfrac{1}{\lambda} =\;\dfrac{\dim(\mathcal{H}_A)}{\left|\text{Im}(\delta_A^{0})\right|\left|\text{Im}(\delta_{\tilde{A}}^{-1})\right|}\;
= \;\dfrac{\left|\text{ker}(\delta_A^0)\right|}{\left|\text{Im}(\delta_{\tilde{A}}^{-1})\right|},
 \end{align}
 Equation (\ref{eq:lambda}) is already very interesting since it relates \(1/\lambda\) to the ground state degeneracy \((GSD)\) of the model restricted to $\mathcal{H}_A$ and for which gauge transformations act in the interior of $A$ only. By replacing this expression into Eq.(\ref{eq:SAlambda}) we are able to state our first result, that the entanglement entropy is given by:
\begin{equation}\label{eq:SAmain}
S_A= \log\left(\dfrac{\left|\text{ker}(\delta_A^0)\right|}{\left|\text{Im}(\delta_{\tilde{A}}^{-1})\right|}\right) = \log \left(GSD_{\tilde{A}} \right).
\end{equation}
We want to highlight that the only requirement we asked for the bipartition is that the simplicial complex $K$ is divided into a subcomplex $K_A$ and its complement. Therefore, this result is very general since it is valid for any higher gauge theory of the type described in Sections \ref{sec:ExampleAHGT}, \ref{sec:higher} and constructed in \cite{higher} and for any arbitrary dimension.

\subsection{Topological Entanglement Entropy}\label{sec:tee}
Having a general result for the bipartition entanglement entropy $S_A$, given by equation (\ref{eq:SAmain}), it is natural to ask if it is possible to extract from it both the area law and the sub-leading, possibly topological terms, explicitly exhibiting its dependency on both the geometry and the topology of the region $A$. Here we give an answer to this question and show how $S_A$ depends on $A$. We demonstrate that the topological contribution to the entanglement entropy comes from topological invariants of the sub-region $A$ and from topological invariants of the entangling surface $\partial A$. These invariants are related to Brown's cohomology groups \cite{Brown}. One can express the cohomology groups of Brown as a product of all cohomology groups of the manifold in question with coefficients in the homology groups of the chain complex of groups (\ref{eq:Gchain2}) (see \cite{Brown} for further details). Thus, for higher gauge theories, the topological entanglement entropy depends on the Betti numbers of $A$ and $\partial A$, but the actual relation is more involved than the case of $1$-gauge theories. 

\par Let's consider equation (\ref{eq:SAmain}). It depends on the quantity $GSD_{\tilde{A}}$, the ground state degeneracy of a model restricted to the region $A$ without gauge transformations at the boundary $\partial A$. This quantity can be rewritten in the following way: multiplying and dividing it by $|\text{Im}(\delta^{-1}_{A})|$, i.e., the order of the group image of $\delta^{-1}_{A}$, we have
\begin{equation}\label{eq:GSD_rew}
	GSD_{\tilde{A}} = GSD_{\tilde{A}}\frac{|\text{Im}(\delta^{-1}_{A})|}{|\text{Im}(\delta^{-1}_{A})|} = \frac{|\text{ker}(\delta^0_A)|}{|\text{Im}(\delta^{-1}_{\tilde{A}})|}\frac{|\text{Im}(\delta^{-1}_{A})|}{|\text{Im}(\delta^{-1}_{A})|} = GSD_A \frac{|\text{Im}(\delta^{-1}_A)|}{|\text{Im}(\delta^{-1}_{\tilde{A}})|},
\end{equation}
where, from theorem \ref{thm:main}, $GSD_A = |\text{ker}(\delta^0_{A})|/|\text{Im}(\delta^{-1}_A)|$ is the ground state degeneracy of the model restricted to region $A$. As explained in section \ref{sec:higher}, $GSD_A = |\text{H}^{0}(C_A, G)|$, where $|\text{H}^{0}(C_A, G)|$ is the order of the zeroth Brown cohomology group of the complex $(C_A, G)$. Therefore, $GSD_A$ is a topological invariant of $A$. However, there are more topological contributions to the entanglement entropy coming from the entangling surface $\partial A$, which can have non-trivial topology. They are hidden in the term $|\text{Im}(\delta^{-1}_{A})|/|\text{Im}(\delta^{-1}_{\tilde{A}})|$, which essentially counts the number of non-equivalent gauge transformations labelled by simplexes at the boundary $\partial A$ of $A$. To see this, first we note that, in the same way we constructed a co-chain complex that encodes an Abelian higher gauge theory defined on the region $A$ but with higher gauge transformations restricted to act on the interior $\tilde{A}$ of $A$, we can also construct a co-chain complex that represents an Abelian higher gauge theory defined on $A$ but with higher gauge transformations that act only on the boundary $\partial A$ of $A$.

\par Let $K_{n,\partial A} = \{x \in K_{n,A}|x \cap \tilde{A} = \emptyset\}$ be the set of $n$-simplices that have no intersection with the interior of $A$. Clearly, $K_{n, \tilde{A}} \cap K_{n, \partial A} = \emptyset$ and $K_{n, \tilde{A}} \cup K_{n, \partial A} = K_{n,A}$. Higher gauge transformations over $A$ are parameterized by elements of the group $\homA^{-1}$. Therefore, the transformations that act over the boundary of $A$ are parameterized by the subgroup $\homdA^{-1} \subset \homA^{-1}$, whose elements are collections of maps $f = \{f_n\}$ such that, for each $0\le n \le d$, $f_n: K_{n, \partial A} \to G_{n+1}$ has support in the boundary of A. Then, we define the restriction $\delta^{-1}_{\partial A} = \delta^{-1}_{A}|_{\partial A}$ of the co-boundary map $\delta^{-1}_{A}$, therefore
\begin{align}\label{eq:redhomcochaindA}
\homdA^{-1} \xrightarrow{\delta_{\partial A}^{-1}} \homA^{0}\xrightarrow{\delta_A^{0}} \homA^1
\end{align}
is a co-chain complex since $\delta_A^{0} \circ \delta_{\partial A}^{-1} = 0$. This co-chain complex defines an Abelian higher gauge theory on $A$ with higher gauge transformations labelled by elements of $\partial A$. From (\ref{eq:redhomcochainA}) and (\ref{eq:redhomcochaindA}), we can establish the following result:

\begin{lem}\label{eq:lem:1}
	\begin{equation}
		\frac{|\text{Im}(\delta^{-1}_A)|}{|\text{Im}(\delta^{-1}_{\tilde{A}})|} = |\text{Im}(\delta^{-1}_{\partial A})|.
	\end{equation}
\end{lem}  
\begin{proof}
	See appendix \ref{app_lemma}.
\end{proof}
\par We can then focus in studying the set $\text{Im}(\delta^{-1}_{\partial A})$. First, note that we can construct, for any $-d \le p \le d$, the group $\homdA^{p} = \bigoplus_{n=0}^{d-1}\text{Hom}(C_{n, \partial A}, G_{n-p})$ and the co-boundary map $\delta^{p}_{\partial A}: \homdA^{p} \to \homdA^{p+1}$, just as we did above. Therefore, we can construct the co-chain complex 
	\begin{align}\label{eq:redhomcochaindA2}
	\homdA^{-d} \xrightarrow{\delta_{\partial A}^{-d}} \cdots \xrightarrow{\delta_{\partial A}^{-2}}  \homdA^{-1} \xrightarrow{\delta_{\partial A}^{0}} \homdA^1 \to \cdots,
	\end{align}
and, for this co-chain complex, we can also define the Brown cohomology groups
\begin{equation}\label{eq:browncohobA}
	\text{H}^p(C_{\partial A}, G) = \text{ker}(\delta^{p}_{\partial A})/\text{Im}(\delta^{p-1}_{\partial A}).
\end{equation}
Now, for any $-d \le p \le d$, we have that
\begin{equation}\label{eq:rec_rel}
|\text{Im}(\delta^{p}_{\partial A})| = \frac{|\homdA^{p}|}{|\text{H}^{p}(C_{\partial A}, G)||\text{Im}(\delta^{(p-1)}_{\partial A})|}.
\end{equation}
Indeed, since $\delta^p_{\partial A}: \homdA^{p} \to \homdA^{p+1}$ is a group homomorphism, the first isomorphism theorem says that
\begin{equation}
	\homdA^{p}/\text{ker}(\delta_{\partial A}^p) \cong \text{Im}(\delta^{p}_{\partial A}),
\end{equation}
which implies that
\begin{equation} \label{eq:rec_re1}
	|\text{Im}(\delta^{p}_{\partial A})| = \frac{|\homdA^{p}|}{|\text{ker}(\delta_{\partial A}^p)|}.
\end{equation}
Now, from equation (\ref{eq:browncohobA}), we have that
\begin{equation}
	|\text{H}^{p}(C_{\partial A}, G)| = \frac{|\text{ker}(\delta^{p}_{\partial A})|}{|\text{Im}(\delta^{p-1}_{\partial A})|},
\end{equation}
so, substituting $|\text{ker}(\delta^{p}_{\partial A})| = |\text{H}^{p}(C_{\partial A}, G)||\text{Im}(\delta^{p-1}_{\partial A})|$ into (\ref{eq:rec_re1}), we get (\ref{eq:rec_rel}). This equation is a recursion relation that allows us to write $|\text{Im}(\delta^{p}_{\partial A})|$ in terms of $|\text{Im}(\delta^{p-1}_{\partial A})|$. In particular, for $p = -1$,
\begin{equation}
	|\text{Im}(\delta^{-1}_{\partial A})| = \frac{|\homdA^{-1}|}{|\text{H}^{-1}(C_{\partial A}, G)||\text{Im}(\delta^{-2}_{\partial A})|}.
\end{equation}
Applying \eqref{eq:rec_rel} once more, we have
\begin{equation}\label{eq:example_recursion}
	|\text{Im}(\delta^{-1}_{\partial A})| = \frac{|\homdA^{-1}||\text{H}^{-2}(C_{\partial A}, G)|}{|\text{H}^{-1}(C_{\partial A}, G)||\homdA^{-2}|}|\text{Im}(\delta^{-3}_{\partial A})|.
\end{equation}
This procedure can be continued until we finally reach 
\begin{equation}
	|\text{Im}(\delta^{-d}_{\partial A})| = \frac{|\homdA^{-d}|}{|\text{H}^{-d}(C_{\partial A}, G)||\text{Im}(\delta^{-(d+1)}_{\partial A})|}.
\end{equation}
This ends the recursion since $|\text{Im}(\delta^{-(d+1)}_{\partial A})| = 1$ because $\homdA^{-(d+1)} = \{0\}$ and $\delta^{-(d+1)}_{\partial A}$ is the inclusion map. Putting all together, we have
\begin{align}\label{eq:im_result}
	|\text{Im}(\delta^{-1}_{\partial A})| = \prod_{p=1}^d|\homdA^{-p}|^{\alpha}|\text{H}^{-p}(C_{\partial A}, G)|^{-\alpha},
\end{align}
where $\alpha = 1$ when $p$ is odd and $\alpha = -1$ when $p$ is even. Thus, we can write $|\text{Im}(\delta^{-1}_{\partial A})|$ as a product of terms that depend on the geometry of the entangling surface $\partial A$, i.e., the $|\homdA^{p}|$ terms, with terms that depend on the topology of $\partial A$, that is, the $|\text{H}^{p}(C_{\partial A}, G)|$ terms, for $-d \le p \le -1$.
\par We can now go back to the entaglement entropy $S_A$. We have that, from equations (\ref{eq:SAmain}), (\ref{eq:GSD_rew}) and lemma \ref{eq:lem:1},
\begin{equation}
	S_A = \text{log}\left(GSD_{\tilde{A}}\right) = \text{log}(GSD_A) + \text{log}\left(|\text{Im}(\delta^{-1}_{\partial A})|\right).
\end{equation}  
From \eqref{eq:im_result}, we have
\begin{align}
	\text{log}\left(|\text{Im}(\delta^{-1}_{\partial A})|\right) = \sum_{p = 1}^d (-1)^{p+1} \text{log}\left(|\homdA^{-p}|\right) + \sum_{p = 1}^d(-1)^p \text{log}\left(|\text{H}^{-p}(C_{\partial A}, G)|\right). \nonumber
\end{align}
Now, as explained in \cite{higher}, here we are considering the case where there are finitely many $n$'s such that $K_n$ is non-empty, and each of these non-empty sets are also finite. Moreover, all groups $G_n$ appearing in the chain complex $(G, \partial^G)$ are finite. Thus, the groups $\hom^p$ are finite for all $p$, and their order is given by
\begin{equation}
	|\hom^p| = \prod_{n=0}^d |\text{Hom}(C_n, G_{n-p})| = \prod_{n=0}^d |G_{n-p}|^{|K_n|}.
\end{equation}
In particular,
\begin{equation}
	|\homdA^{-p}| = \prod_{n=0}^{d-1} |G_{n+p}|^{|K_{n,\partial A}|},
\end{equation}
and thus
\begin{align}
	\text{log}\left(|\homdA^{-p}|\right) = \text{log}\left(\prod_{n=0}^{d-1} |G_{n+p}|^{|K_{n,\partial A}}|\right)	= \sum_{n=0}^{d-1}|K_{n,\partial A}|\text{log}\left(|G_{n+p}|\right).
\end{align}
Now, using the isomorphism (\ref{eq:BrownIsomorphism}), we can write
\begin{align}
\text{log}(GSD_A) = \text{log}\left(|\text{H}^0(C_A, G)|\right)
= \sum_{n=0}^d\text{log}\left(|\text{H}^n(C_A, H_n(G))|\right),
\end{align}
and
\begin{align}
\sum_{p = 1}^d(-1)^p \text{log}\left(|\text{H}^{-p}(C_{\partial A}, G)|\right)
=  \sum_{n=0}^{d-1}\sum_{p = 1}^d(-1)^p \text{log}\left(|\text{H}^{n}(C_{\partial A}, H_{n+p}(G))|\right).
\end{align}
Therefore, the entanglement entropy $S_A$ can be written as
\begin{align}\label{eq:sab_satop}
	S_A = S_{\partial A} + S_{\text{Topo}},
\end{align}
where
\begin{equation}\label{eq:saba}
	S_{\partial A} = \sum_{n=0}^{d-1}\sum_{p = 1}^d (-1)^{p+1}|K_{n,\partial A}|\text{log}\left(|G_{n+p}|\right)
\end{equation}
is the "area law" term, that is, the term which explicitly depends only on the geometry of $\partial A$, and
\begin{align}\label{eq:satopo}
S_{\text{Topo}} = \sum_{n=0}^d\text{log}\left(|\text{H}^n(C_A, H_n(G))|\right) + \sum_{n=0}^{d-1}\sum_{p = 1}^d(-1)^p \text{log}\left(|\text{H}^{n}(C_{\partial A}, H_{n+p}(G))|\right)
\end{align}
is the topological entanglement entropy, i.e., the term that explicitly depends on the topology of both $A$ and $\partial A$. To calculate it, we must be able to compute cohomology groups with coefficients in the homology groups of the chain complex $(G,\partial^G)$. Cohomology with coefficients is related to the usual integral homology through the \textit{universal coefficient theorem} (see \cite{Hatcher} for a general reference), which states that, for any Abelian group $D$ and any $0\le n\le d$,

\begin{align}\label{eq:uni_coef_thm}
	\text{H}^n(C, D) \cong \text{Hom}(\text{H}_n(C), D) \oplus \text{Ext}^1(\text{H}_{n-1}(C), D),
\end{align}
where $\text{H}_n(C)$ is the homology group of order $n$ with integer coefficientes of the chain complex $C = (C,\partial^C)$. The $\text{Ext}^1$ term is related to the torsion part of $H_{n}(C)$, and it is trivial whenever $H_n(C)$ is free. Writting $H_n(C) \cong \mathbb{Z}^{\beta_n} \oplus T_i$, where $\mathbb{Z}^{\beta_n}$ corresponds to the free part of $H_n(C)$, $\beta_n$ being the Betti number of order $n$ of the chain complex $C$, and $T_i$ is its torsion part, we can see that indeed there is a relation between $S_{\text{Topo}}$ and the higher Betti numbers of both the region $A$ and its boundary $\partial A$. In particular, when we consider manifolds with torsion-free homology groups, the $\text{Ext}^1$ term in (\ref{eq:uni_coef_thm}) is trivial and we have as a result
\begin{align}
	|\text{H}^n(C, D)| = |\text{Hom}(\text{H}_n(C), D)| = |\text{Hom}(\mathbb{Z}^{\beta_n}, D)| = |D|^{\beta_n}.
\end{align}
However, the dependency on the Betti numbers of $A$ and $\partial A$ can be non-trivial when the homology groups have torsion. In this case we must compute the $\text{Ext}^1$ term. In section \ref{sec:Examples}, we illustrate our results by calculating the entanglement entropy of several examples coming from higher gauge theories, which include the familiar Quantum Double Models in their Abelian versions.

\section{Examples}\label{sec:Examples}
In this section we calculate the entanglement entropy of the model shown in \S \ref{sec:ExampleAHGT}. We will use this example to demonstrate how the topology of both $A$ and $\partial A$ affect the entanglement entropy. For $1$-gauge theories (Abelian Quantum Double Models), we recover well-known general results for any dimension $d$.

\subsection{1-Gauge Theories}\label{EE-TC}
We consider the general Abelian $1$-gauge theory, also known as the Abelian Quantum Double model, but now in any dimension $d$. This theory is defined by the chain complexes in figure \ref{fig:ChainMaps1Gauge}. This example can be seen as a particular case of the $d$-dimensional version of the model introduced in \ref{ex:12}, where the only non-trivial group is $G_1$.

\begin{figure}[h!]
	\begin{center}
		\begin{adjustbox}{max size={.7\textwidth}{.7\textheight}}
			\begin{tikzpicture}
			\matrix (m) [matrix of math nodes,row sep=2em,column sep=3em,minimum width=2em]
			{  0 & C_{d} & \cdots & C_{2} & C_{1} & C_{0} & 0 \\
				  &  &  & 0 & G_{1} & 0 &  \\};
			\path[-stealth]
			(m-1-1) edge node [above] {$$} (m-1-2)
			(m-1-2) edge node [above] {$\partial^C_{d} $} (m-1-3)
			(m-1-3) edge node [above] {$\partial^C_3 $} (m-1-4)
			(m-1-4) edge node [above] {$\partial^C_2$} (m-1-5)
			(m-1-5) edge node [above] {$\partial^C_1$} (m-1-6)
			edge node [right] {$f_1$} (m-2-5)
			(m-1-6) edge node [above] {$$} (m-1-7)  
			(m-2-4) edge node [below] {$ $} (m-2-5)
			(m-2-5) edge node [below] {$ $} (m-2-6);
			\path[dotted,->]   
			(m-1-4) edge node [left] {$\; m_{2}$} (m-2-5)
			(m-1-6) edge node [right] {$\; t_{0}$} (m-2-5);
			\end{tikzpicture}
		\end{adjustbox}
	\end{center}
	\caption{\label{fig:ChainMaps1Gauge} We obtain the chain complexes that define Abelian $1$-gauge theories in $d$ dimensions by setting $G_0=G_2=0$ in example \ref{ex:12}. The relevant maps are the configurations associated to the links $f_1$, the gauge transformations $t_0$, and the 1-holonomies $m_2$, as explained in \ref{ex:12-chain}.  }
\end{figure}

The geometrical chain complex, generated by the simplicial complex $K = \bigcup_{n=0}^dK_{n}$, corresponds to the $d$-dimensional manifold $M$ in which the theory is defined. Let's consider a bipartition of this manifold into two regions, $A$ and $B$, where $A$ is a $d$-dimensional closed immersed submanifold of $M$, as in the general case discussed in section \ref{sec:bipart}. This corresponds to splitting the simplicial complex $K$ as $K =\bigcup_{n=0}^{d} K_{n,A} \cup K_{n, B}$, where $\bigcup_{n=0}^d K_{n,A}$ is a subcomplex of $K$. Therefore, following the general procedure shown in section \ref{sec:bipart}, we can construct for region $A$ the chain complexes shown in figure \ref{fig:ChainMaps1GaugeA}.   

\begin{figure}[h!]
	\begin{center}
		\begin{adjustbox}{max size={.7\textwidth}{.7\textheight}}
			\begin{tikzpicture}
			\matrix (m) [matrix of math nodes,row sep=2em,column sep=3em,minimum width=2em]
			{  0 & C_{d,A} & \cdots & C_{2,A} & C_{1,A} & C_{0,A} & 0 \\
				 &  &  & 0 & G_{1} & 0 &  \\};
			\path[-stealth]
			(m-1-1) edge node [above] {$$} (m-1-2)
			(m-1-2) edge node [above] {$\partial^C_{d,A} $} (m-1-3)
			(m-1-3) edge node [above] {$\partial^C_{3,A} $} (m-1-4)
			(m-1-4) edge node [above] {$\partial^C_{2,A}$} (m-1-5)
			(m-1-5) edge node [above] {$\partial^C_{1,A}$} (m-1-6)
			edge node [right] {$f_1$} (m-2-5)
			(m-1-6) edge node [above] {$$} (m-1-7)   
			(m-2-4) edge node [below] {$ $} (m-2-5)
			(m-2-5) edge node [below] {$ $} (m-2-6);
			\path[dotted,->]    
			(m-1-4) edge node [left] {$\; m_{2}$} (m-2-5)
			(m-1-6) edge node [right] {$\; t_{0}$} (m-2-5);
			\end{tikzpicture}
		\end{adjustbox}
	\end{center}
	\caption{\label{fig:ChainMaps1GaugeA} Chain complexes that define Abelian $1$-gauge theories in $d$ dimensions for a subregion $A$, as explained in \ref{ex:12-chain}. The maps are the ones explained in \ref{fig:ChainMaps1Gauge} restricted to the subcomplex $K_A$.}
\end{figure}
From these chain complexes, as was demonstrated in sections \ref{sec:bipart}, \ref{sec:redrho} and \ref{sec:tee}, the following cochain complexes are well defined for this theory:
\begin{align}\label{eq:cochain1gaugeA}
\homA^{-d} \xrightarrow{\delta_A^{-d}} \cdots \xrightarrow{\delta_{A}^{-2}}  \homA^{-1} \xrightarrow{\delta_{A}^{0}} \homA^1 \to \cdots,
\end{align}
\begin{align}\label{eq:cochain1gaugebA}
\homdA^{-d} \xrightarrow{\delta_{\partial A}^{-d}} \cdots \xrightarrow{\delta_{\partial A}^{-2}}  \homdA^{-1} \xrightarrow{\delta_{\partial A}^{0}} \homdA^1 \to \cdots.
\end{align}
We can then apply equation (\ref{eq:sab_satop}) to find the entanglement entropy $S_A$ of this theory. First, let's compute $S_{\partial A}$, given by equation (\ref{eq:saba}). We have that
\begin{equation}\label{eq:sa_1gauge}
	S_{\partial A} = \sum_{n=0}^{d-1}\sum_{p=1}^d(-1)^{p+1}|K_{n,\partial A}|\text{log}(|G_{n+p}|).
\end{equation}
This quantity depends only on the simplicial complex of the boundary of $A$ and on the higher gauge groups of the theory in question. For $1$-gauge, there is only one non-trivial gauge group $G_1$, all other groups being equal to $\{0\}$, as is shown in figure \ref{fig:ChainMaps1GaugeA}. Therefore, since the order of trivial groups is equal to one, the only term which will survive in the double sum in equation (\ref{eq:sa_1gauge}) is the $n=0$, $p=1$ term, and we have as a result
\begin{equation}
	S_{\partial A} = |K_{0,\partial A}|\text{log}(|G_1|).
\end{equation} 
Thus, we have indeed an "area law" term, because $S_{\partial A}$ is proportional to the number of vertices in the boundary of $A$, which is essentially the size of the boundary. We are only left to calculate $S_{\text{Topo}}$. From equation (\ref{eq:satopo}), we know that
\begin{equation}\label{eq:satopo1gauge}
	S_{\text{Topo}} = \sum_{n=0}^d\text{log}\left(|\text{H}^n(C_A, H_n(G))|\right) + \sum_{n=0}^{d-1}\sum_{p = 1}^d(-1)^p \text{log}\left(|\text{H}^{n}(C_{\partial A}, H_{n+p}(G))|\right).
\end{equation}
We will leave the topology of $A$ and $\partial A$ unspecified for a moment and proceed as far as possible. For a 1-gauge theory, the non-trivial piece of the chain complex $(G, \partial^G)$ is given by 
\begin{align}
	0 \xrightarrow{\partial^G_2} G_1 \xrightarrow{\partial^G_1} 0,
\end{align} 
and the only non-trivial homology group related to this chain complex is \[\text{H}_1(G_1) = \text{ker}(\partial^G_1)/\text{Im}(\partial^G_2) =  G_1.\] Therefore, in the first term of (\ref{eq:satopo1gauge}) the only non-zero contribution to the sum comes from $n=1$, while for the second term of (\ref{eq:satopo1gauge}) the only non-zero contribution to the sum comes from $n=0, p=1$. Hence, the topological entanglement entropy in Abelian $1$-gauge theories is given by
\begin{equation}
	S_{\text{Topo}} = \text{log}\left(|\text{H}^1(C_A, G_1)|\right) - \text{log}\left(|\text{H}^{0}(C_{\partial A}, G_1)|\right).
\end{equation}
Then, to compute $S_{\text{Topo}}$ we need to calculate the order of the cohomology groups $\text{H}^1(C_A, G_1)$ and $\text{H}^{0}(C_{\partial A}, G_1)$, with coefficients in $G_1$. To do so, we employ the universal coefficient theorem for cohomology (\ref{eq:uni_coef_thm}), which states that
\begin{align}
	\text{H}^n(C, G_1) \cong \text{Hom}(\text{H}_n(C), G_1) \oplus \text{Ext}^1(\text{H}_{n-1}(C), G_1),
\end{align}
for every $0\le n \le d$, where $\text{H}_n(C)$ is the $n$th homology group with integer coefficients of the chain complex $C$. Here we only need the $n = 0$ and $n = 1$ cases. For $n = 0$, there is no $\text{Ext}$ term \cite{Hatcher}, so
\begin{equation}
	\text{H}^0(C_{\partial A}, G_1) \cong \text{Hom}(\text{H}_0(C_{\partial A}), G_1).
\end{equation}
We have also that $\text{H}_0(C_{\partial A}) \cong \mathbb{Z}^{\beta_0(\partial A)}$, where $\beta_0({\partial A})$ is the zeroth Betti number of $\partial A$. Therefore,
\begin{equation}
	\text{H}^0(C_{\partial A}, G_1) \cong \text{Hom}(\mathbb{Z}^{\beta_0({\partial A})}, G_1)
\end{equation}
and hence
\begin{equation}
	|\text{H}^0(C_{\partial A}, G_1)| = |\text{Hom}(\mathbb{Z}^{\beta_0({\partial A})}, G_1)| = |G_1|^{\beta_0({\partial A})}.
\end{equation}
Likewise, the $\text{Ext}$ term is trivial for the $n=1$ case, and we can write
\begin{equation}
	\text{H}^1(C_{A}, G_1) \cong \text{Hom}(\text{H}_1(C_{A}), G_1).
\end{equation}
So in general, the topological entanglement entropy for Abelian $1$-gauge theories reads
\begin{equation}
	S_{\text{Topo}} = \text{log}\left(|\text{Hom}(\text{H}_1(C_{A}), G_1)|\right) - \beta_{0}({\partial A})\text{log}\left(|G_1|\right),
\end{equation}
and it depends on the number of non-contractible curves one can draw over region $A$, as well as on the number of connected components of $\partial A$. To proceed further, we need to specify the topology of $A$. For example, if $A$ has the topology of a $d$-dimensional ball, its first homology group is trivial and hence we get that $S_{\text{Topo}} = - \beta_{0}({\partial A})\text{log}\left(|G_1|\right)$, i.e., it is only sensitive to the number of connected components of the entangling surface $\partial A$. If we take $A = T^d$, where $T^d$ is the $d$-torus $T^d = (S^1)^d$, we have that $\text{H}_1(C_{A}) \cong \mathbb{Z}^{\beta_1(A)}$, where $\beta_1(A)$ is the first Betti number of $A$ which, in this case, is equal to $\beta_1(A) = d$. Then $|\text{Hom}(\text{H}_1(C_{A}), G_1)| = |G_1|^{\beta_1(A)}$ and \[S_{\text{Topo}} = (\beta_1(A) - \beta_{0}(\partial A))\text{log}\left(|G_1|\right).\]

\par It is important to mention that we could also build theories with degrees of freedom attached to other components of the lattice, such as its $2$-simplices (plaquettes). Doing so, we have a particular case of a 2-gauge theory. As we will see in the following discussion, this makes the topological entanglement entropy depend on higher Betti numbers. An example of such model is the $4$-dimensional Toric Code \cite{dennis2002topological}. Note that in \cite{dennis2002topological} it is not called a 2-gauge theory. We do so in order to be consistent with our formalism. Let's then consider an Abelian $2$-gauge theory in $d$-dimensions with degrees of freedom living at the $2$-simplices of the lattice. The chain complexes describing this model is presented in figure \ref{fig:ChainMaps1Gaugefaces}

\begin{figure}[h!]
	\begin{center}
		\begin{adjustbox}{max size={.7\textwidth}{.7\textheight}}
			\begin{tikzpicture}
			\matrix (m) [matrix of math nodes,row sep=2em,column sep=3em,minimum width=2em]
			{  0 & C_{d} & \cdots &C_3& C_{2} & C_{1} & C_{0} & 0 \\
				&  &  & 0 & G_{2} & 0 &  &  \\};
			\path[-stealth]
			(m-1-1) edge node [above] {$$} (m-1-2)
			(m-1-2) edge node [above] {$\partial^C_{d} $} (m-1-3)
			(m-1-3) edge node [above] {$$} (m-1-4)
			(m-1-4) edge node [above] {$\partial^C_3 $} (m-1-5)
			(m-1-5) edge node [above] {$\partial^C_2$} (m-1-6)
			edge node [right] {$f_2$} (m-2-5)
			(m-1-6) edge node [above] {$\partial^C_1$} (m-1-7)
			(m-1-7) edge node [above] {$$} (m-1-8)  
			(m-2-5) edge node [below] {$ $} (m-2-6)
			(m-2-4) edge node [below] {$ $} (m-2-5);
			\path[dotted,->] 
			(m-1-4) edge node [left] {$\; m_3 $} (m-2-5)   
			(m-1-6) edge node [right] {$\; t_{1}$} (m-2-5);
			\end{tikzpicture}
		\end{adjustbox}
	\end{center}
	\caption{\label{fig:ChainMaps1Gaugefaces} Chain complexes that define Abelian $1$-gauge theories in $d$ dimensions with degrees of freedom living at the faces of the lattice. $f_2$ are the configurations, $t_1$ are the gauge transformations labeled by links, and there is also $m_3$  as the 2-holonomies.}
\end{figure}

We can repeat exactly the same steps we did for calculating the entanglement entropy of the $1$-gauge model obtained from the chain complexes in figure \ref{fig:ChainMaps1Gauge}. Hence, we can skip this discussion and right away apply the formula (\ref{eq:sab_satop}) to calculate the entanglement entropy of this model. Here, since the only non-trivial group is $G_2$, the only non-zero terms in the sum that define $S_{\partial A}$ are the $p=1, n=1$ and $p=2, n=0$ terms, so $S_{\partial A}$ is given by
\begin{equation}
	S_{\partial A} = |K_{1,\partial A}|\text{log}(|G_2|) - |K_{0,\partial A}|\text{log}(|G_{2}|) = \left(|K_{1,\partial A}| - |K_{0,\partial A}|\right)\text{log}(|G_2|).
\end{equation}   
Note that this term vanishes for a $3$-dimensional region $A$ with periodic boundary conditions. Now, since the only non-trivial piece of the group chain complex shown in figure (\ref{fig:ChainMaps1Gaugefaces}) is \[0 \xrightarrow{\partial^G_3} G_2 \xrightarrow{\partial^G_2} 0,\] the only non-trivial homology group associated to this complex is \[\text{H}_2(G) = \text{ker}(\partial^G_2)/\text{Im}(\partial^G_3) = G_2.\] Thus, the non-zero contribution to the first term of $S_{\text{Topo}}$ comes from the term with $n=2$, while the non-zero contributions to the second term of $S_{\text{Topo}}$ come from the terms with $p=1,n=1$ and $p=2, n=0$. Therefore, the topological entanglement entropy of this model is given by
\begin{align}\label{eq:1gauge_G2}
	S_{\text{Topo}} = \text{log}\left(|\text{H}^2(C_A, G_2)|\right) + \text{log}\left(|\text{H}^0(C_{\partial A}, G_2)|\right) - \text{log}\left(|\text{H}^1(C_{\partial A}, G_2)|\right).  
\end{align}
Again, using the universal coefficient theorem (\ref{eq:uni_coef_thm}) to calculate the cohomologies, we first note that there is no $\text{Ext}$ term for $n = 0$ and it is trivial for $n=1$. Therefore,
\begin{align}
	|\text{H}^0(C_{\partial A}, G_2)| = |\text{Hom}(H_0(C_{\partial A}), G_2)| = |G_2|^{\beta_0(\partial A)}, \label{eq:coh_dA1} \\
	|\text{H}^1(C_{\partial A}, G_2)| = |\text{Hom}(H_1(C_{\partial A}), G_2)|, \label{eq:coh_dA2}
\end{align} 
where we used again that $H_0(C_{\partial A}) \cong \mathbb{Z}^{\beta_0(\partial A)}$. However, applying the theorem to $|\text{H}^2(C_A, G_2)|$, we have that 
\begin{align}\label{eq:hom_ext}
	|\text{H}^2(C_A, G_2)| = |\text{Hom}(\text{H}_2(C_A), G_2)||\text{Ext}^1(\text{H}_{1}(C_A), G_2)|,
\end{align}
and we can't say much about $\text{Ext}^1(\text{H}_{1}(C_A), G_2)$ without knowing the homology group $\text{H}_{1}(C_A)$. For example, whenever $\text{H}_{1}(C_A)$ is free, wich happens when we, for instance, choose $A$ to be a $d$-dimensional ball or a $d$-torus, $\text{Ext}^1(\text{H}_{1}(C_A), G_2) = \{0\}$ and we have that \[|\text{H}^2(C_A, G_2)| = |\text{Hom}(\text{H}_2(C_A), G_2)|,\] and the topological entanglement entropy now is related to the number of non-contractible curves and surfaces one can draw over $\partial A$ and $A$, respectively, and therefore $S_{\text{Topo}}$ depends on higher Betti numbers of both $A$ and $\partial A$. 

To give an example where the $\text{Ext}$ term is non-trivial, take the $4$-dimensional manifold $A = \mathbb{R}P^3 \times [0,1]$, i.e., the product of the real projective space with the unit interval. It is a manifold whose boundary is $\partial A = \mathbb{R}P^3 \cup \mathbb{R}P^3$. The homology groups of $\mathbb{R}P^3$ are
\begin{align}
	\text{H}_n(\mathbb{R}P^3) = \begin{cases}
	\mathbb{Z}, &\mbox{if}\quad n = 0 \quad\mbox{or}\quad n = 3, \\
	\mathbb{Z}_2, &\mbox{if}\quad n = 1,\\
	0, &\mbox{otherwise}.
	\end{cases}
\end{align}
Let's compute the second homology group of $A$. By K\"{u}nneth's theorem \cite{maclane}, we have that
\begin{align}
	\text{H}_2(\mathbb{R}P^3 \times [0,1]) \cong \bigoplus_{i+j = 2}\text{H}_i(\mathbb{R}P^3) \otimes \text{H}_j([0,1]) \nonumber \\
	= \text{H}_0(\mathbb{R}P^3) \otimes \text{H}_2([0,1]) \oplus \text{H}_2(\mathbb{R}P^3) \otimes \text{H}_0([0,1]) \oplus \text{H}_1(\mathbb{R}P^3) \otimes \text{H}_1([0,1]) \nonumber \\
	\cong (\mathbb{Z} \otimes \{0\}) \oplus (\{0\} \otimes \mathbb{Z}) \oplus (\mathbb{Z}_2 \otimes \{0\}), \nonumber
\end{align} 
that is,
\begin{align}
	\text{H}_2(\mathbb{R}P^3 \times [0,1]) \cong \mathbb{Z} \oplus \mathbb{Z} \oplus \mathbb{Z}_2.
\end{align}
The first homology group of $A = \mathbb{R}P^3 \times [0,1]$ can be computed in the same way:
\begin{align}
	\text{H}_1(\mathbb{R}P^3 \times [0,1]) \cong \bigoplus_{i+j = 1}\text{H}_i(\mathbb{R}P^3) \otimes \text{H}_j([0,1]) \nonumber \\
	 = \text{H}_0(\mathbb{R}P^3) \otimes \text{H}_1([0,1]) \oplus \text{H}_1(\mathbb{R}P^3) \otimes \text{H}_0([0,1]) \nonumber \\
	 \cong (\mathbb{Z} \otimes \{0\}) \oplus (\mathbb{Z}_2 \otimes \Z), \nonumber
\end{align}
that is,
\begin{align}
	\text{H}_1(\mathbb{R}P^3 \times [0,1]) \cong \mathbb{Z}\oplus \mathbb{Z}_2 \otimes \Z.
\end{align}
Then, from equation \eqref{eq:hom_ext}, we have
\begin{align}
	|\text{H}^2(C_A, G_2)| = |\text{Hom}(\mathbb{Z} \oplus \mathbb{Z} \oplus \mathbb{Z}_2, G_2)||\text{Ext}^1(\mathbb{Z}\oplus \mathbb{Z}_2 \otimes \Z, G_2)| \nonumber \\
	= |G_2|^4|\text{Ext}^1(\mathbb{Z}\oplus \mathbb{Z}_2  \otimes \Z, G_2)|, \nonumber 
\end{align}
and $\Z_2\otimes \Z \cong \Z_2$ \cite{maclane} and that \[\text{Ext}^1(\mathbb{Z}\oplus \mathbb{Z}_2, G_2) \cong \text{Ext}^1(\mathbb{Z}, G_2) \oplus \text{Ext}^1(\mathbb{Z}_2, G_2) \cong \text{Ext}^1(\mathbb{Z}_2, G_2),\]
because $\text{Ext}^1(\mathbb{Z}, G_2) = \{0\}$ \cite{Hatcher}, we have
\begin{align}
	|\text{H}^2(C_A, G_2)| = |G_2|^4||G_2/2G_2|,
\end{align}
where we used that $\text{Ext}^1(\mathbb{Z}_n, G) \cong G/nG$, for any $n \in \mathbb{N}$ and any Abelian group $G$ \cite{Hatcher}. Now, the order of the cohomology groups \eqref{eq:coh_dA1} and \eqref{eq:coh_dA2} of the boundary $\partial A = \mathbb{R}P^3 \cup \mathbb{R}P^3$ are
\begin{align}
	|\text{H}^0(C_{\partial A}, G_2)| = |G_2|^2, \\
	|\text{H}^1(C_{\partial A}, G_2)| = |\text{Hom}(\mathbb{Z}_2 \oplus \mathbb{Z}_2, G_2)| = |G_2|^4,
\end{align}
and thus, the topological entanglement entropy is given by
\begin{align}
	S_{\text{Topo}} = \text{log}(|G_2|^4|G_2/2G_2|) + \text{log}(|G_2|^2) - \text{log}(|G_2|^4),
\end{align}
i.e.,
\begin{align}
	S_{\text{Topo}} = \text{log}(|G_2/2G_2|) + 2\text{log}(|G_2|).
\end{align}

Expanding the results found in the literature \cite{Grover11, zheng2018structure}, here we demonstrated that the topological entanglement entropy depends not only on the topology of the entangling surface $\partial A$, but also on the topological properties of the bulk region $A$. Moreover, the Betti numbers of $A$ and $\partial A$ are not the only information needed to obtain the topological entropy. For some models, it may depend also on \textit{torsion} properties of the sub-region $A$ and its boundary, captured mainly by the $\text{Ext}^1$ functor. 

\subsection{(2D) 0,1-Gauge Theories}\label{EE-01_Z4Z2}

Let's consider a $0,1$-gauge theory, a particular case of the one exhibited in \S\ref{ex:12-chain} with $G_2 = 0$. For simplicity, we focus in the $2$-dimensional case, but the discussion presented here can immediately be extended to any dimension $d$. The chain complexes defining this theory are shown in figure \ref{fig:ChainMapsr-0,1}.

\begin{figure}[h!]
	\begin{center}
		\begin{adjustbox}{max size={.7\textwidth}{.7\textheight}}
			\begin{tikzpicture}
			\matrix (m) [matrix of math nodes,row sep=2em,column sep=3em,minimum width=2em]
			{  0 & C_{2} & C_{1} & C_{0} & 0 \\
				& 0  & G_{1} & G_{0} & 0 \\};
			\path[-stealth]
			(m-1-1) edge node [above] {$ $} (m-1-2)
			(m-1-2) edge node [above] {$\partial^C_{2} $} (m-1-3)
			(m-1-3) edge node [above] {$\partial^C_{1}$} (m-1-4)
			edge node [left] {$f_1$} (m-2-3)
			(m-1-4) edge node [above] {$$} (m-1-5)       
			edge node [right] {$f_0$} (m-2-4)
			(m-2-2) edge node [below] {$$} (m-2-3)
			(m-2-3) edge node [below] {$\partial^G_{1} $} (m-2-4)
			(m-2-4) edge node [below] {$ $} (m-2-5);
			\path[dotted,->]     
			(m-1-2) edge node [left] {$\; m_{2}$} (m-2-3)
			(m-1-4) edge node [right] {$\; t_{0}$} (m-2-3);
			\end{tikzpicture}
		\end{adjustbox}
	\end{center}
	\caption{\label{fig:ChainMapsr-0,1} Chain complexes for the 0,1-gauge model in \ref{EE-01_Z4Z2}. The configurations for the links are determined by $f_1$, while for the vertices by $f_0$. The gauge transformations are given by $t_0$, while we have the 1-holonomies $m_2$ and 0-holonomies $m_1$ not shown in the figure.}
\end{figure}

Again, we divide the lattice into two regions, $A$ and $B$, where for $A$ we have a subcomplex $\bigcup_n^d K_{n,A}$. Following the general procedure shown in \S\ref{sec:bipart}, we can construct for region $A$ the chain complex $(C_A,\partial^C_A)$ and find that the entanglement entropy $S_A$ is given by equation (\ref{eq:sab_satop}). Let's calculate first $S_{\partial A}$. It is straightforward to see that the only non-zero contribution to the sum in (\ref{eq:saba}) is given by the $n = 0, p = 1$ term. So we have that
\begin{align}
	S_{\partial A} = |K_{0,\partial A}|\text{log}(|G_1|),
\end{align}
and it is again an "area law". Now, consider the $(G,\partial^G)$ chain complex of this theory:
\begin{align}
	0 \xrightarrow{\partial^G_2} G_1 \xrightarrow{\partial^G_1} G_0 \xrightarrow{\partial^G_0} 0.
\end{align} 
The non-trivial homology groups associated to this chain complex are
\begin{align}
	\text{H}_0(G) = \text{ker}(\partial^G_0)/\text{Im}(\partial^G_1), \label{eq:homol01_0} \\
	\text{H}_1(G) = \text{ker}(\partial^G_1)/\text{Im}(\partial^G_2). \label{eq:homol01_1}
\end{align}
So, the non-zero contributions to the topological entanglement entropy (\ref{eq:satopo}) are
\begin{align}
	S_{\text{Topo}} = \text{log}\left(|\text{H}^0(C_A, \text{H}_0(G))|\right) + \text{log}\left(|\text{H}^1(C_A, \text{H}_1(G))|\right) - \text{log}\left(|\text{H}^0(C_{\partial A}, \text{H}_1(G))|\right). 
\end{align}  
This result is true for any Abelian finite groups $G_0$ and $G_1$. However, now we have to specify these groups in order to calculate the homology groups (\ref{eq:homol01_0}) and (\ref{eq:homol01_1}). Thus, let's first consider the model with $G_0 = \mathbb{Z}_2=\{0,1\}$ and $G_1 = \mathbb{Z}_4=\{0,1,2,3\}$. The map $\partial_1^G: \mathbb{Z}_4 \to \mathbb{Z}_2$ is defined by $\partial^G_1(1) = 1$. The homology groups (\ref{eq:homol01_0}) and (\ref{eq:homol01_1}) are thus
\begin{align}
	\text{H}_0(G) = \text{ker}(\partial^G_0)/\text{Im}(\partial^G_1) = \mathbb{Z}_2/\mathbb{Z}_2 = \{0\}, \\
	\text{H}_1(G) = \text{ker}(\partial^G_1)/\text{Im}(\partial^G_2) = \mathbb{Z}_2.
\end{align}
The topological entanglement entropy of this model is thus
\begin{align}
	S_{\text{Topo}} = \text{log}\left(|\text{H}^1(C_A, \mathbb{Z}_2)|\right) - \text{log}\left(|\text{H}^0(C_{\partial A}, \mathbb{Z}_2)|\right).
\end{align}
Note that this result is equal to the $G_1 = \mathbb{Z}_2$ $1$-gauge theory. Again, this model exhibits the same topological properties as the Toric Code.
\par We can also consider the model where $G_0 = \mathbb{Z}_2$ and $G_1 = \mathbb{Z}_2$, with $\partial_1^G: \mathbb{Z}_2 \to \mathbb{Z}_2$ being the identity map. In this case, both $\text{H}_0(G)$ and $\text{H}_1(G)$ are equal to the trivial group, and the topological entanglement entropy is equal to zero, confirming the non-topological nature of this model.   

\subsection{(3D) 1,2-Gauge Theories}\label{EE-12}

This time we consider the case of the 1,2-gauge theory (\S\ref{ex:12} and \S\ref{ex:12-chain} with $G_0 = 0$). Although here we treat the $3$-dimensional case, the procedure below can right away be extended to arbitrary dimensions $d$. The chain complexes that define this theory are reproduced here in figure \ref{fig:ChainMaps-1,2r}.

\begin{figure}[h!]
	\begin{center}
		\begin{adjustbox}{max size={.7\textwidth}{.7\textheight}}
			\begin{tikzpicture}
			\matrix (m) [matrix of math nodes,row sep=2em,column sep=3em,minimum width=2em]
			{  0 & C_{3} & C_{2} & C_{1} & C_{0} & 0 \\
				& 0 & G_{2} & G_{1} & 0 &\\};
			\path[-stealth]
			(m-1-1) edge node [above] {$$} (m-1-2)
			(m-1-2) edge node [above] {$\partial^C_3 $} (m-1-3)
			(m-1-3) edge node [above] {$\partial^C_{2} $} (m-1-4)
			edge node [left] {$f_2$} (m-2-3)
			(m-1-4) edge node [above] {$\partial^C_{1}$} (m-1-5)
			edge node [right] {$f_1$} (m-2-4)
			(m-1-5) edge node [above] {$$} (m-1-6)          		
			(m-2-2) edge node [below] {$ $} (m-2-3)
			(m-2-3) edge node [below] {$\partial^G_{2} $} (m-2-4)
			(m-2-4) edge node [below] {$ $} (m-2-5);
			\path[dotted,->]   
			(m-1-2) edge node [left] {$\;m_{3}$} (m-2-3)  
			(m-1-4) edge node [right] {$\;t_{1}$} (m-2-3) 
			(m-1-5) edge node [right] {$\; t_{0}$} (m-2-4);
			\end{tikzpicture}
		\end{adjustbox}
	\end{center}
	\caption{\label{fig:ChainMaps-1,2r} Chain complexes for the 1,2-gauge model. The configurations maps are $f_2$ for faces, $f_1$ for links, the gauge transformations are the $t_0$, $t_1$; while the 2-holonomies are given by $m_3$, the 1-holonomies by $m_2$.}
\end{figure}

Once more, we divide the lattice into two regions, $A$ and $B$, where region $A$ is such that we have a subcomplex $K_{A}$. From the general procedure shown in \S\ref{sec:bipart}, we can construct for region $A$ the chain complex $(C_A,\partial^C_A)$ and then find that the entanglement entropy $S_A$ is given by equation (\ref{eq:sab_satop}). To compute it for this model, let's calculate first $S_{\partial A}$. We see that the only non-zero contributions to the sum in (\ref{eq:saba}) are the ones given by the $n = 0, p = 1$, $n=1, p=1$ and $n=0,p=2$ terms. Therefore
\begin{align}
	S_{\partial A} = |K_{0,\partial A}|\text{log}\left(|G_1|\right) + |K_{1,\partial A}|\text{log}\left(|G_2|\right) - |K_{0,\partial A}|\text{log}\left(|G_2|\right) \\
	= |K_{0,\partial A}|\text{log}\left({|G_1|}\right) + (|K_{1,\partial A}|-|K_{0,\partial A}|)\text{log}\left(|G_2|\right).
\end{align}
Note that, if we were dealing with a two dimensional system with periodic boundary conditions, the terms proportional to $\text{log}(|G_2|)$ would cancel out, as the number of links and vertices is the same, and we would have the same result as the one found in the $1$-gauge case.
\par Now, to calculate the topological entanglement entropy $S_{\text{Topo}}$, given by equation (\ref{eq:satopo}), first we consider the chain complex $(G,\partial^G)$ of this model:
\begin{align}
	0 \xrightarrow{\partial_3^G} G_2 \xrightarrow{\partial_2^G} G_1 \xrightarrow{\partial_1^G} 0.
\end{align}
The non-trivial homology groups associated to this chain complex are as follows:
\begin{align}
	\text{H}_1(G) = \text{ker}(\partial^G_1)/\text{Im}(\partial^G_2), \label{eq:homol_12gauge1} \\
	\text{H}_2(G) = \text{ker}(\partial^G_2)/\text{Im}(\partial^G_3). \label{eq:homol_12gauge2}
\end{align}
Thus, $S_{\text{Topo}}$ is given by
\begin{align}
	S_{\text{Topo}} = \text{log}\left(|\text{H}^1(C_A, \text{H}_1(G))|\right) + \text{log}\left(|\text{H}^2(C_A, \text{H}_2(G))|\right) - \text{log}\left(|\text{H}^0(C_{\partial A}, \text{H}_1(G))|\right) + \\
	- \; \text{log}\left(|\text{H}^1(C_{\partial A}, \text{H}_2(G))|\right) + \text{log}\left(|\text{H}^0(C_{\partial A}, \text{H}_2(G))|\right).
\end{align}
To proceed further, let's choose the model where $G_2 = \mathbb{Z}_4$, $G_1 = \mathbb{Z}_2$ and $\partial_2^G:\mathbb{Z}_4 \to \mathbb{Z}_2$ such that $\partial_2^G(1) = 1$. We can then calculate the homology groups (\ref{eq:homol_12gauge1}) and (\ref{eq:homol_12gauge2}) to be
\begin{align}
	\text{H}_1(G) = \{0\}, \\
	\text{H}_2(G) = \mathbb{Z}_2.
\end{align}
Therefore, the topological entanglement entropy of this model is
\begin{align}\label{eq:12Gauge_Stopo}
	S_{\text{Topo}} = \text{log}\left(|\text{H}^2(C_A, \mathbb{Z}_2)|\right) - \text{log}\left(|\text{H}^1(C_{\partial A}, \mathbb{Z}_2)|\right) + \text{log}\left(|\text{H}^0(C_{\partial A}, \mathbb{Z}_2)|\right).
\end{align}
This formula is equal to the $2$-gauge (\ref{eq:1gauge_G2}) case, with degrees of freedom living in the plaquettes of the lattice and $G_2 = \mathbb{Z}_2$. Note that we never used the fact that we are dealing with a $3$-dimensional system to derive equation (\ref{eq:12Gauge_Stopo}), which means that the same result holds for dimension $d\ge 3$. Then, if we for example take $d = 4$, this choice of groups for the $1,2$-gauge model imposes that its long-range entanglement characteristics, detected by the topological entanglement entropy, are the same as the $4D$ Toric Code one studied in \cite{dennis2002topological}. Other choices of groups may generate $1,2$-gauge models with more unusual behaviors.

\subsection{(4D) 1,2,3-Gauge Theories}\label{EE-123}

Now we consider the $1,2,3$-gauge theory in four dimensions. We do this to show that our formalism allows us to readily shift from a $3$-dimensional presented in \ref{ex:12} case to a 4D, our formalism can be extended to any arbitrary dimension $d$. The chain complexes that define this theory is shown in figure \ref{fig:ChainMaps-1,2,3r}.
\begin{figure}[h!]
	\begin{center}
		\begin{adjustbox}{max size={.7\textwidth}{.7\textheight}}
			\begin{tikzpicture}
			\matrix (m) [matrix of math nodes,row sep=2em,column sep=3em,minimum width=2em]
			{  0&C_4 & C_{3} & C_{2} & C_{1} & C_{0} & 0 \\
				&0 & G_3 & G_{2} & G_{1} & 0 &\\};
			\path[-stealth]
			(m-1-1) edge node [above] {$$} (m-1-2)
			(m-1-2) edge node [above] {$\partial^C_4 $} (m-1-3)
			(m-1-3) edge node [above] {$\partial^C_{3} $} (m-1-4)
			edge node [right] {$f_3$} (m-2-3)
			(m-1-4) edge node [above] {$\partial^C_{2}$} (m-1-5)
			edge node [right] {$f_2$} (m-2-4)
			(m-1-5) edge node [above] {$\partial^C_{1}$} (m-1-6)  
			edge node [right] {$f_1$} (m-2-5)     
			(m-1-6) edge node [above] {} (m-1-7)     		
			(m-2-2) edge node [below] {$\partial^G_3$} (m-2-3)
			(m-2-3) edge node [below] {$\partial^G_{2} $} (m-2-4)
			(m-2-4) edge node [below] {$\partial^G_1 $} (m-2-5)
			(m-2-5) edge node [below] {$ $} (m-2-6);
			\path[dotted,->]
			(m-1-2) edge node [left] {$\;m_{4}$} (m-2-3)    
			(m-1-6) edge node [right] {$\;t_{0}$} (m-2-5)    
			(m-1-4) edge node [right] {$\;t_{2}$} (m-2-3) 
			(m-1-5) edge node [right] {$\; t_{1}$} (m-2-4);
			\end{tikzpicture}
		\end{adjustbox}
	\end{center}
	\caption{\label{fig:ChainMaps-1,2,3r} Chain complexes for the 1,2,3-gauge model.}
\end{figure}

We again divide the lattice into two regions, $A$ and $B$, where region $A$ is such that we have a subcomplex $\bigcup_{n=0}^dK_{n,A}$. The general procedure shown in \S\ref{sec:bipart} allows us to construct for region $A$ the chain complex $(C_A,\partial^C_A)$ and thus we can find the entanglement entropy $S_A$ using equation (\ref{eq:sab_satop}). To calculate it for the $1,2,3$-gauge model, let's study first the term $S_{\partial A}$. We see that the only non-zero contributions to the sum in (\ref{eq:saba}) are the ones given by the $p = 1, n = 0,1,2$, $p=2, n = 0,1$ and $p=3,n=0$ terms. Therefore
\begin{align}
	S_{\partial A} = |K_{0,\partial A}|\text{log}\left(|G_1|\right) + |K_{1,\partial A}|\text{log}\left(|G_2|\right)  + |K_{2,\partial A}|\text{log}\left(|G_3|\right) +  \nonumber\\ 
	- \; |K_{0,\partial A}|\text{log}\left(|G_2|\right) - |K_{1,\partial A}|\text{log}\left(|G_3|\right) + |K_{0,\partial A}|\text{log}\left(|G_3|\right), \nonumber
\end{align}
that is,
\begin{align}
	S_{\partial A} = |K_{0,\partial A}|\text{log}\left(\frac{|G_1||G_3|}{|G_2|}\right) + |K_{1,\partial A}|\text{log}\left(\frac{|G_2|}{|G_3|}\right) + |K_{2,\partial A}|\text{log}\left(|G_3|\right). 
\end{align}
To calculate the topological entanglement entropy (\ref{eq:satopo}), we first consider the chain complex $(G,\partial^G)$ of the $1,2,3$-gauge model:
\begin{align}
	0 \xrightarrow{\partial_4^G} G_3 \xrightarrow{\partial_3^G} G_2 \xrightarrow{\partial_2^G} G_1 \xrightarrow{\partial_1^G} 0.
\end{align}
The non-trivial homology groups associated to it are the following ones:
\begin{align}
	\text{H}_1(G) = \text{ker}(\partial^G_1)/\text{Im}(\partial^G_2), \label{eq:homol_123gauge1} \\
	\text{H}_2(G) = \text{ker}(\partial^G_2)/\text{Im}(\partial^G_3), \label{eq:homol_123gauge2} \\
	\text{H}_3(G) = \text{ker}(\partial^G_3)/\text{Im}(\partial^G_4). \label{eq:homol_123gauge3}
\end{align}
Hence, $S_{\text{Topo}}$ is given by
\begin{align}
	S_{\text{Topo}} = \text{log}\left(|\text{H}^1(C_A, \text{H}_1(G))|\right) + \text{log}\left(|\text{H}^2(C_A, \text{H}_2(G))|\right) + \text{log}\left(|\text{H}^3(C_A, \text{H}_3(G))|\right) + \nonumber \\ -\;  \text{log}\left(|\text{H}^0(C_{\partial A}, \text{H}_1(G))|\right) - \text{log}\left(|\text{H}^1(C_{\partial A}, \text{H}_2(G))|\right) - \text{log}\left(|\text{H}^2(C_{\partial A}, \text{H}_3(G))|\right) + \nonumber \\ + \;
	\text{log}\left(|\text{H}^0(C_{\partial A}, \text{H}_2(G))|\right) + \text{log}\left(|\text{H}^1(C_{\partial A}, \text{H}_3(G))|\right) - \text{log}\left(|\text{H}^0(C_{\partial A}, \text{H}_3(G))|\right).
\end{align}
Note that the topological entanglement entropy does not change from the 3D to the 4D case, because the homological groups are the same
. 
To give a more concrete example, let's consider the case where $G_1 = G_2 = G_3 = \mathbb{Z}_4$ and the homomorphisms $\partial^G_3(1) = \partial^G_2(1) = 1$. The lattice is a discretization of a solid ball $S^3$. In this case, we can calculate the homology groups (\ref{eq:homol_123gauge1}), (\ref{eq:homol_123gauge2}) and (\ref{eq:homol_123gauge3}). They are
\begin{align}
	\text{H}_1(G) = \mathbb{Z}_4/\mathbb{Z}_2 \cong \mathbb{Z}_2, \\
	\text{H}_2(G) = \{0\}, \\
	\text{H}_3(G) = \mathbb{Z}_2.
\end{align}
Therefore, the topological entanglement entropy of this model is 
\begin{align}
	S_{\text{Topo}} = \text{log}\left(|\text{H}^1(C_A, \mathbb{Z}_2)|\right) + \text{log}\left(|\text{H}^3(C_A, \mathbb{Z}_2)|\right) - 2\text{log}\left(|\text{H}^0(C_{\partial A}, \mathbb{Z}_2)|\right) + \nonumber \\ - \; \text{log}\left(|\text{H}^2(C_{\partial A}, \mathbb{Z}_2)|\right) + \text{log}\left(|\text{H}^1(C_{\partial A}, \mathbb{Z}_2)|\right).
\end{align}
Let's use the universal coefficient theorem to compute these cohomology groups. We have that, as before,
\begin{align}
	|\text{H}^0(C_{\partial A}, \mathbb{Z}_2)| = |\text{Hom}(\text{H}_0(C_{\partial A}), \mathbb{Z}_2)| = 2^{\beta_0(\partial A)}, \\
	|\text{H}^1(C_{\partial A}, \mathbb{Z}_2)| = |\text{Hom}(\text{H}_1(C_{\partial A}), \mathbb{Z}_2)|, \\
	|\text{H}^1(C_A, \mathbb{Z}_2)| = |\text{Hom}(\text{H}_1(C_A), \mathbb{Z}_2)|.
\end{align}
Since in this case $\partial A = S^2$, we have that $\beta_0(\partial A) = 1$, $\text{H}_1(C_{\partial A}) = \text{H}_1(C_A) = \{0\}$. Therefore, $|\text{H}^1(C_{\partial A}, \mathbb{Z}_2)| = |\text{H}^1(C_A, \mathbb{Z}_2)| = 1$. Now, 
\begin{align}
	|\text{H}^2(C_{\partial A}, \mathbb{Z}_2)| = |\text{Hom}(\text{H}_2(C_{\partial A}), \mathbb{Z}_2)||\text{Ext}^1(\text{H}_{1}(C_{\partial A}), \mathbb{Z}_2)|, \\
	|\text{H}^3(C_{A}, \mathbb{Z}_2)| = |\text{Hom}(\text{H}_3(C_{A}), \mathbb{Z}_2)||\text{Ext}^1(\text{H}_{2}(C_{A}), \mathbb{Z}_2)|
\end{align}
and, since $\text{H}_1(C_{\partial A}) = \text{H}_2(C_A) = \{0\}$, the $\text{Ext}$ terms are trivial. So, with $\text{H}_2(C_{\partial A}) \cong \text{H}_3(C_{A}) \cong \mathbb{Z}$, we have that $|\text{H}^2(C_{\partial A}, \mathbb{Z}_2)| = |\text{H}^3(C_{A}, \mathbb{Z}_2)| = 2$. Hence, the topological entanglement entropy of this model is given by
\begin{align}
	S_{\text{Topo}} = -2\text{log}(2). 
\end{align}  
We see that, although we defined the model over a manifold with trivial topology and the ground state degeneracy of this model does not exhibit a topological dependency, the topological entanglement entropy is different from zero, indicating the presence of long-range entanglement.

\section{Conclusions}\label{sec:Remarks}
The paper carried out the calculation of the entanglement entropy for all Abelian higher gauge theories in a comprehensive way. Furthermore we could separate the entropy into the topological information and the geometrical one. We started by making a review of the kind of models we treated. Then we described them in very general terms, as introduced in \cite{higher}. The calculation followed from the definition of the density matrix $\rho$ as being proportional to the ground state projector, see \eqref{rho}. To obtain the reduced density matrix we considered a bipartition of the simplicial complex $K$ into a subcomplex $K_A$ and its complement. The partial trace over the unknown region was used to obtain the reduced density matrix $\rho_A$, which  included operators that were exclusively supported in $K_A$, see \eqref{eq:rhoA2}. From the Von Neumann entropy formula we derived the entanglement entropy and showed that it could be naturally interpreted as the ground state degeneracy of the same model but restricted to the subcomplex $K_A$, see \eqref{eq:SAmain}. In this sense, we mapped the problem of calculating the entanglement entropy of a higher gauge theory to a problem of counting the flat edge states of the theory restricted to region \(A\). Then, we further divided this restricted ground state degeneracy into two contributions, one comming from the bulk region $A$ and the other comming from its boundary $\partial A$, and we showed that this splitting allows us to write the entanglement entropy as a sum of two terms \eqref{eq:sab_satop}: one being the \textit{area law}, i.e., a term depending only on the geometry of the entangling surface $\partial A$, and the other being the \textit{topological entanglement entropy}, a term depending on the topological properties of both $A$ and $\partial A$. 

We demonstrated a formula for the topological entanglement entropy $S_{\text{Topo}}$ in terms of the cohomology groups with coefficients in the homology groups of the complex \eqref{eq:Gchain2}. The universal coefficient theorem can be applied to give a formula for $S_{\text{Topo}}$ in terms of the integral homology groups of the manifold in question, which in turn can be used to express the topological entropy in terms of the Betti numbers of the underlying space and its boundary. However, our equations show that, even in regular $1$-gauge theories (Abelian Quantum Double models), $S_{\text{Topo}}$ can depend on torsion properties of the manifold.

\appendix

\section{Trace of Local Operators}\label{sec:app2}
In this appendix we show how taking the partial trace of the ground state projector, or any product of projection operators of the theory, implies in Eq.(\ref{eq:rhoA1}). 

We begin by writing the density matrix, \(\rho\), using the local decomposition of \(\mathcal{A}_0\) and \(\mathcal{B}_0\) (see \cite{higher} for a detailed account on this). The local decomposition yields
\begin{align*}
\mathcal{A}_0 = \prod_{n=0}^{d} \prod_{x\in K_n} A_{n,x}, \; \text{and}\quad \mathcal{B}_0 = \prod_{n=0}^{d} \prod_{x\in K_n} B_{n,x},
\end{align*}
such that the density matrix of Eq.(\ref{eq:rho}) can be written as
\begin{align*}
\rho = \dfrac{1}{GSD} \left(\prod_{n=0}^{d} \prod_{x\in K_n} A_{x}\right)\left(\prod_{n=0}^{d} \prod_{x\in K_n} B_{x}\right).
\end{align*}
This form is convenient for taking the partial trace as the operators are now labelled by simplices $x \in K_n$ for $0\leq n \leq d$. This allows the identification of the operators that act exclusively on region $A$ from the operators that act on both $\partial A$ and $B$, in order to get the terms that survive the partial trace.
Therefore, the reduced density matrix is written as
\begin{align}\label{eq:rhoAlocal}
\rho_A = \text{Tr}_B(\rho)= \;\text{Tr}_B\left(\prod_n \prod_{x \in K_n}A_{n,x}\prod_{y \in K_n}B_{n,y}\right).
\end{align}
Before proceeding with the calculation of the above partial trace, we will introduce a property that will let us evaluate the partial trace rather straightforwardly.

\begin{prop}\label{prop:traceless}
Let \(x,y \in K_n\), be $n$-simplices for \(0\leq n\leq d\). The local operators, \(A_{n,x}, B_{n,y}:\mathcal{H}\rightarrow\mathcal{H}\), are traceless unless they act  trivially (as the identity operator \(\mathbb{1}_\mathcal{H}\)).
\end{prop}
\begin{proof}
Let \(\{\ket{f}\}\) be a basis of $\mathcal{H}$, with \(f \in \hom^0\). We start by taking the trace of the local operator $A_{n,x}$:
\begin{align*}
\text{Tr}\left(A_{n,x}\right) = & \sum_{f} \bra{f} A_{n,x} \ket{f} = \dfrac{1}{\left|G_{n+1}\right|}\sum_{f}\, \sum_{g \in G_{n+1}} \bra{f} A_{e[n,x,g]} \ket{f}.
\end{align*}
From \eqref{eq:At}, the action of \(A_{n,x}\) on a basis state consists in general on a shift of basis elements, which yields
\begin{align*}
\text{Tr}\left(A_{n,x}\right) =&  \dfrac{1}{\left|G_{n+1}\right|}\sum_{f}\, \sum_{g \in G_{n+1}} \braket{f|f+ \delta^{-1}(e[n,x,g])|f}.
\end{align*}
From the last expression, by using the orthogonality of the basis, it is clear that the only non-null term in the sum occurs only when \(g=e \in G_{n+1}\), the identity element. Thus, we have:
\begin{align*}
\text{Tr}\left(A_{n,x}\right) =& \dfrac{\text{Tr}\left(\mathbb{1}\right)}{\left|G_{n+1}\right|} = \dfrac{\dim(\mathcal{H})}{\left|G_{n+1}\right|}.
\end{align*}
 Similarly, for the trace of local holonomy measurement operators, \(B_{n,y}\), we have:
\begin{align*}
\text{Tr}\left(B_{n,y}\right) = & \sum_{f} \braket{f| B_{x}|f} = \dfrac{1}{\left|G_{n-1}\right|}\sum_{f}\, \sum_{r \in \hat{G}_{n-1}} \bra{f} B_{\hat{e}[n,y,r]} \ket{f}. 
\end{align*}
Using \eqref{eq:Bm} the above expression can be written as:
\begin{align*}
\text{Tr}\left(B_{y}\right) =\;  \dfrac{1}{\left|G_{n-1}\right|}\sum_{f}\, \sum_{r \in \hat{G}_{n-1}} \langle r, \delta^0 f_n(y)\rangle\braket{f|f}\;
=\;  \dfrac{1}{\left|G_{n-1}\right|}\sum_{f}\, \sum_{r \in \hat{G}_{n-1}} \langle r, \delta^0 f_n(y)\rangle\langle \hat{e}, \delta^0 f_n(y)\rangle\braket{f|f},
\end{align*}
where in the last line we used the fact that \(\langle \hat{e}, g\rangle=1, \, \forall g \in G_{n-1}\) and \(\hat{e} \in \hat{G}_{n-1}\), the trivial representation. From the orthogonality relations of characters \cite{hall, barut,serre, james}, we note that:
\begin{align*}
\sum_{f}\langle r, \delta^0 f_n(y)\rangle\langle \hat{e}, \delta^0 f_n(y)\rangle = \delta(e,f_n(y)), 
\end{align*}
which implies that the trivial representation term is the only one that has non-zero trace, since it acts as the identity operator.
\begin{align*}
\text{Tr}\left(B_{y}\right) = & \dfrac{\left|\mathcal{H}\right|}{\left|G_{n-1}\right|}
\end{align*}
\end{proof}
 This result can naturally be extended to products of such operators to show that the only term that survives the trace is the one that acts trivially on region $B$. This allows us to express  the reduced density matrix, \(\rho_A\) of Eq.(\ref{eq:rhoAlocal}) in terms of operators that act only in region \(A\). 
 
 In this case,  Proposition \ref{prop:traceless} implies that any operator (or product of several) that is different from \(\mathbb{1}_B\), the identity operator in \(\mathcal{H}_B\), will have vanishing trace. In particular, local gauge transformations \(A_{x}\) will survive the trace if and only if \(x \in K_{n,\tilde{A}}\), where \(\tilde{A}\) is the interior of \(A\) \footnote{Local gauge transformations are labeled by simplices \(x \in K_n\) and they act on the gauge fields at the co-boundary, \(\partial^{\ast}(x)\). In particular, gauge transformations located at \(x \in K_{n,{\partial A}}\), the boundary of \(A\), also act on \(B\). Thus, they do not contribute to the trace.} as in Def. \ref{def:intA}. On the other hand, local holonomy measurement operators \(B_{y}\) will survive the trace if and only if \(y \in K_{n,A}\) which corresponds to the entire region $A$. Consequently, the reduced density matrix is:
\begin{align*}
\rho_A =\text{Tr}_B(\mathbb{1}_B) \prod_{n} \prod_{x \in K_{n,\tilde{A}}}A_{x}\prod_{y \in K_{n,A}}B_{y} .
\end{align*}  
From which we write Eq. (\ref{eq:rhoA1}).

\section{Auxiliary Isomorphism}\label{sec:app3}
In this appendix, we prove the equality \(\left|\text{Im}(\delta^0)\right|=\left|\text{Im}(\delta_{1})\right|\) that allowed us to relate the dimension of the Hilbert space $\mathcal{H}$ and the ground state degeneracy \(GSD\) through: 
\[GSD \,|\text{Im}(\delta^{-1})||\text{Im}(\delta_{1})| =\dim(\mathcal{H})= \dim(\mathcal{H}_A) \dim(\mathcal{H}_B).\]
In order to do so, we will show that there is a well defined bijection between \(\text{ker}(\delta_1)\) and \(\hom^1 / \text{Im}(\delta^0)\) from which the result follows.

Let \(A, B\) be two finite Abelian groups and \(\phi:A \rightarrow B\) a homomorphism between them. Consider also \(\hat{A}=\Hom (A, U(1))\) and \(\hat{B}=\Hom (B, U(1))\) their corresponding unitary irreducible representations, let \(\hat{\phi}:\hat{B}\rightarrow \hat{A}\) be the homomorphism between representations induced by \(\phi\) via
\[\hat{\phi}(\beta) := \beta \circ \phi, \]
where $\beta \in \hat{B}$ is an irrep of $B$. 

\begin{prop}\label{prop:Iso}
The subgroups \(\text{ker}\ \hat{\phi}\) and \(\frac{B}{\text{Im}(\phi)}\) are isomorphic.
\end{prop}
\begin{proof}
We will split the proof in two parts, in the first half of the proof we show that there is a well defined map between  \(\text{ker}\ \hat{\phi}\) and \(\frac{B}{\text{Im}(\phi)}\) and then we show that its inverse is also well defined, which turns the maps into a bijection.
\begin{enumerate}
\item Note that an irreducible representation $\beta\in \text{ker}\hat{\phi}$ if and only if \(\text{Im}\phi \subset \text{ker}\beta\). This allows us to construct the following commuting diagram:
\begin{equation}\label{com-diag}
\xymatrix{
B\ar[d]_\pi\ar[r]^\beta&U(1)\\
\frac{B}{\im\phi}\ar@{-->}[ur]_{\beta'}
}\end{equation}
where \(\pi: B \rightarrow \frac{B}{\text{Im}\phi}\) is the canonical projection sending \(b\in B \)  into its corresponding equivalence class \([b] \in \frac{B}{\text{Im}\phi} \)
Furthermore, \(\beta^\prime \in \Hom(\frac{B}{\text{Im}\phi}, U(1))\) is unique and defined as: 
\[\beta^{\prime}([b]):= \beta(b) \]
notice that \(\beta^\prime\) is well defined within equivalence classes since \(\text{Im}\phi \subset \text{ker}\beta\). To see this, consider \(b^\prime \neq b \in [b]\), this means that \(b-b^{\prime} \in \text{Im}\phi \subset \text{ker}\beta\), therefore:
\begin{align*}
\beta(b-b^\prime) = 1, \quad & \Rightarrow   \beta(b)\beta(b^\prime)^{-1}= 1, \\
& \Rightarrow \beta(b) = \beta(b^\prime)= \beta^{\prime}([b]).
\end{align*}
This is, we have shown that given an irrep \(\beta \in \text{ker} \hat{\phi}\) then there is a unique morphism \(\beta^{\prime} \in \Hom(\frac{B}{\text{Im}\phi}, U(1))\).

We now need to show that the converse also holds, to this intent, consider \(\beta^{\prime}:\frac{B}{\text{Im}\phi} \rightarrow U(1)\). Recall that \(\text{Im}\phi \subset \text{ker}\beta\). Observe also that \(\beta\) is the only map for which the diagram in \ref{com-diag} commutes.

Thus, we have shown that given a $\beta^{\prime} \in \Hom(\frac{B}{\text{Im}\phi}, U(1))$ there is a unique $\beta= \beta^{\prime}\circ \pi  \in \text{ker}\hat{\phi}$.

\item Now we carry on showing that the map above is in fact a bijection and it defines an isomorphism. Let \(\iota\) be the map:
\begin{align*}
\iota: \text{ker}\hat{\phi} & \longrightarrow \Hom\left(\frac{B}{\text{Im}\phi}, U(1)\right), \\
\beta\quad  &\,\,\mapsto \quad \beta^{\prime},
\end{align*}
where \(\beta^{\prime}([b]):= \beta(b)\). 
Let now, \(\kappa\), be the map:
\begin{align*}
\kappa:\Hom\left(\frac{B}{\text{Im}\phi}, U(1)\right)  & \longrightarrow \text{ker}\hat{\phi} , \\
\beta^{\prime}\quad  &\,\,\mapsto \quad \beta:=\beta^{\prime}\circ \pi,
\end{align*}
where \(\pi: B \rightarrow \frac{B}{\text{Im}\phi}\) is the canonical projection that sends \(b \in B\) into its corresponding equivalence class \([b] \in \frac{B}{\text{Im}\phi} \). Notice that \(\kappa = \iota^{-1}\), since:
\begin{align*}
\left(\kappa \circ \iota\right)(\beta)(b) = \kappa(\beta^\prime)(b) 
										  = (\beta^\prime \circ \pi)(b)
										 = \beta^\prime ([b])
									     = \beta(b).
\end{align*}
Therefore, the map \(\iota\) is a bijection. To prove that it defines an isomorphism we only need to check for its compatibility with the group operation in \( \ker \hat{\phi}\). This is, given \(\beta_1, \beta_2 \in \ker \hat{\phi}\), we want to show that \(\iota(\beta_1\cdot\beta_2)=\iota(\beta_1)\cdot\iota(\beta_2)\). 

So, consider \(b \in B\) and , \([b]\in \frac{B}{\text{Im}\phi} \):
\begin{align*}
\iota(\beta_1\cdot\beta_2)([b]) = (\beta_1\cdot\beta_2)^{\prime}([b]) = (\beta_1\cdot\beta_2)(b) = \beta_1(b)\cdot\beta_2(b) = \iota(\beta_1)\cdot\iota(\beta_2).
\end{align*}
Hence, \(\text{ker}\hat{\phi} \simeq \Hom\left(\frac{B}{\text{Im}\phi}, U(1)\right)\).
\end{enumerate}
\end{proof}
In particular, as a result of the above proposition, it is true that, for \(A, B\) finite groups:
\begin{equation}\label{eq:Iso}
\left|\text{ker}\hat{\phi}\right| = \left|\Hom\left(\frac{B}{\text{Im}\phi}, U(1)\right)   \right| = \dfrac{\left|B \right|}{\left|\text{Im} \phi \right|},
\end{equation}
where in the last step we used the fact that all groups are Abelian. We are one step away from our goal which can be stated as the following proposition

\begin{prop}
Let \(\phi:A \rightarrow B\) be a homomorphism between finite Abelian groups. Moreover, let \(\hat{\phi}: \hat{B} \rightarrow \hat{A}\) its dual morphism. Then, 
\[\left|\text{Im}\,\phi\right| = \left|\text{Im}\,\hat{\phi}\right|.\]
\end{prop}
\begin{proof}
From Prop. \ref{prop:Iso}, we know that: \(\left|\text{ker}\hat{\phi}\right| = \dfrac{\left|B \right|}{\left|\text{Im} \phi \right|}\).
Now, applying the First Isomorphism Theorem \cite{basicalg} on \(\hat{\phi}: \hat{B} \rightarrow \hat{A}\), we know that: \(\hat{B}/\text{ker}\hat{\phi} \simeq \text{Im}\hat{\phi}\), from which we can write:
\begin{align*}
\dfrac{|\hat{B}|}{|\text{ker}\hat{\phi}|} = \left|\text{Im}\, \hat{\phi} \right|,
\end{align*}
recall that \(|\hat{B}|= |B|\) since we are dealing with Abelian groups. Replacing Eq. (\ref{eq:Iso}) into the above one, we get:
\begin{align*}
\left|\text{Im}\,\phi\right| =  |\text{Im}\, \hat{\phi} |.
\end{align*}
\end{proof} 

\section{Proof of lemma \ref{eq:lem:1}}\label{app_lemma}
We start with a proposition:
\begin{prop}
	\begin{equation} 
	|\text{ker}(\delta_{A}^{-1})| = |\text{ker}(\delta^{-1}_{\tilde{A}})||\text{ker}(\delta^{-1}_{\partial A})|.
	\end{equation}
\end{prop}  
\begin{proof}
	Consider $\text{ker}(\delta^{-1}_A)$, a subgroup of $\homA^{-1}$.
	From the definitions of $\homtA^{-1}$, $\delta_{\tilde{A}}^{-1}$, $\homdA^{-1}$ and $\delta_{\partial A}^{-1}$, it is clear that \[\text{ker}(\delta^{-1}_{\tilde{A}}) = \{f \in \homtA^{-1}| \delta_{\tilde{A}}^{-1}f = 0\}\] and \[\text{ker}(\delta_{\partial A}^{-1}) = \{f \in \homdA^{-1}|\delta_{\partial A}^{-1}f = 0\}\] are subgroups of $\text{ker}(\delta_{A}^{-1})$. We introduce the following equivalence relation on $\text{ker}(\delta^{-1}_A)$: let $f,f' \in \text{ker}(\delta^{-1}_A)$,
	\begin{equation}
	f \sim f' \Leftrightarrow f'-f = g \in \text{ker}(\delta^{-1}_{\tilde{A}}).
	\end{equation}
	That is, two collections of maps $f$ and $f'$ in $\text{ker}(\delta^{-1}_A)$ are equivalent if they differ by a collection of maps in $\text{ker}(\delta^{-1}_{\tilde{A}})$, i.e., maps with support in $K_{\tilde{A}} = \cup_{n = 0}^d K_{n,\tilde{A}}$ that are also sent to the trivial map by the co-boundary operator. This equivalence relation defines the quotient group $\text{ker}(\delta_{A}^{-1})/\text{ker}(\delta^{-1}_{\tilde{A}})$. Define the map $\phi: \text{ker}(\delta_{A}^{-1})/\text{ker}(\delta^{-1}_{\tilde{A}}) \to \text{ker}(\delta_{\partial A}^{-1})$, which for a class $[f] \in \text{ker}(\delta_{A}^{-1})/\text{ker}(\delta^{-1}_{\tilde{A}})$, $\phi([f]) = f|_{\partial A} \in \text{ker}(\delta^{-1}_{\partial A})$, where $f|_{\partial A} = \{(f|_{\partial A})_n\}$ is a collection of maps in $\homdA^{-1}$ such that
	\begin{eqnarray}
	(f|_{\partial A})_n(x) = \begin{cases}
	f_n(x), &\mbox{if}\quad x \in K_{n, \partial A}, \\
	0, &\mbox{\text{otherwise}}.
	\end{cases}  
	\end{eqnarray}
	So, $\phi$ is a map that sends any class $[f] \in \text{ker}(\delta_{A}^{-1})/\text{ker}(\delta^{-1}_{\tilde{A}})$ to the restriction $f|_{\partial A}$ to $\partial A$ of one of its representatives. This map does not depend on the choice of representative of a class. Indeed, let $[f] \in \text{ker}(\delta_{A}^{-1})/\text{ker}(\delta^{-1}_{\tilde{A}})$ and choose two representatives $f', f'' \in [f]$, $f' \neq f''$. We could have that $\phi([f]) = f'|_{\partial A}$ and $\phi[f] = f''|_{\partial A}$. But there is $g \in \text{ker}(\delta^{-1}_{\tilde{A}})$ such that $f'-f'' = g$, so $f'|_{\partial A} - f''|_{\partial A} = g|_{\partial A} = 0$, because $g \in \homtA^{-1}$. Therefore, $f'|_{\partial A} = f''|_{\partial A}$. Moreover, $\phi$ is an isomorphism. To see this, first take $[f] \in \text{ker}(\phi)$. So, $\phi([f]) = 0 \Leftrightarrow f|_{\partial A} = 0$, which means that $f$ is a collection of trivial maps. Therefore, $\text{ker}(\phi) = \{0\}$ and $\phi$ is injective. The map $\phi$ is also surjective, because if we take a map $g \in \text{ker}(\delta^{-1}_{\partial A})$, it is a collection of maps which are zero everywhere except in the boundary of $A$ and it can be obtained by applying $\phi$ in the class $[g] \in \text{ker}(\delta_{A}^{-1})/\text{ker}(\delta^{-1}_{\tilde{A}})$. Therefore, $\phi$ is bijective. Now, let $[f] , [g] \in \text{ker}(\delta_{A}^{-1})/\text{ker}(\delta^{-1}_{\tilde{A}})$. The sum $[f] + [g]$ is given by $[f] + [g] = [f+g]$. We have then $\phi([f] + [g]) = \phi([f+g]) = (f+g)|_{\partial A} = f|_{\partial A} + g|_{\partial A} = \phi([f]) + \phi([g])$, for any representatives $f \in [f]$ and $g \in [g]$. So, $\phi$ is also a homomorphism and therefore it is indeed an isomorphism. Thus, we have indeed that
	\begin{equation} 
	\text{ker}(\delta_{A}^{-1})/\text{ker}(\delta^{-1}_{\tilde{A}}) \cong \text{ker}(\delta^{-1}_{\partial A})
	\end{equation}
	which implies that
	\begin{equation}\label{eq:ker_dec}
	|\text{ker}(\delta_{A}^{-1})| = |\text{ker}(\delta^{-1}_{\tilde{A}})||\text{ker}(\delta^{-1}_{\partial A})|.
	\end{equation}
	\end{proof}
	Now, since $K_{n, \tilde{A}} \cap K_{n, \partial A} = \emptyset$ and $K_{n, \tilde{A}} \cup K_{n, \partial A} = K_{n,A}$ for any $n=0,...,d$, we can write the $n$-chain group $C_{n,A}$ as a direct sum of subgroups $C_{n,A} = C_{n,\tilde{A}} \oplus C_{n,\partial A}$, i.e., every $c = \sum_{x \in K_{n,A}}c(x)x \in C_{n,A}$ can be written as $c = \sum_{x \in K_{n,\tilde{A}}}c(x)x + \sum_{x \in K_{n,\partial A}}c(x)x$. This implies that any homomorphism $f_n:C_{n,A}\to G_{n+1}$, where $G_{n+1}$ is some arbitrary finite Abelian group, can be written as $f_{n}:C_{n,\tilde{A}} \oplus C_{n,\partial A} \to G_{n+1}$. Thus,
	\begin{equation}
		\text{Hom}(C_{n,A}, G_{n+1}) = \text{Hom}(C_{n,\tilde{A}} \oplus C_{n,\partial A}, G_{n+1}). \nonumber 
	\end{equation}
	There is a natural isomorphism \cite{lang, maclane}
	\begin{equation}
		\text{Hom}(C_{n,\tilde{A}} \oplus C_{n,\partial A}, G_{n+1}) \cong \text{Hom}(C_{n,\tilde{A}}, G_{n+1}) \oplus \text{Hom}(C_{n,\partial A}, G_{n+1}), \nonumber 
	\end{equation} 
	thus,
	\begin{equation}
		\text{Hom}(C_{n,A}, G_{n+1}) \cong \text{Hom}(C_{n,\tilde{A}}, G_{n+1}) \oplus \text{Hom}(C_{n,\partial A}, G_{n+1}),
	\end{equation}
	this implies that
	\begin{equation}\label{hom_iso}
		\homA^{-1} \cong \homtA^{-1} \oplus \homdA^{-1},
	\end{equation}
	and thus
	\begin{equation}\label{eq:hom_dec}
		|\homA^{-1}| = |\homtA^{-1}||\homdA^{-1}|.
	\end{equation}
	Then, from the first isomorphism theorem, we have that
	\begin{eqnarray}
	\homA^{-1}/\text{ker}(\delta_{A}^{-1}) \cong \text{Im}(\delta^{-1}_A), \\
	\homtA^{-1}/\text{ker}(\delta^{-1}_{\tilde{A}}) \cong \text{Im}(\delta^{-1}_{\tilde{A}}), \\
	\homdA^{-1}/\text{ker}(\delta^{-1}_{\partial A}) \cong \text{Im}(\delta^{-1}_{\partial A}),
	\end{eqnarray}
	so 
	\begin{eqnarray}
	|\text{Im}(\delta^{-1}_A)| = \frac{|\homA^{-1}|}{|\text{ker}(\delta_{A}^{-1})|}, \label{eq:imA} \\
	|\text{Im}(\delta^{-1}_{\tilde{A}})| = \frac{|\homtA^{-1}|}{|\text{ker}(\delta_{\tilde{A}}^{-1})|}, \label{eq:imtA} \\
	|\text{Im}(\delta^{-1}_{\partial A})| = \frac{|\homdA^{-1}|}{|\text{ker}(\delta_{\partial A}^{-1})|}. \label{eq:imbA}
	\end{eqnarray}
	Thus, dividing (\ref{eq:imA}) by (\ref{eq:imtA}), we have finally that
	\begin{eqnarray}
	\frac{|\text{Im}(\delta^{-1}_A)|}{|\text{Im}(\delta^{-1}_{\tilde{A}})|} = \frac{|\homA^{-1}|}{|\text{ker}(\delta_{A}^{-1})|}\frac{|\text{ker}(\delta_{\tilde{A}}^{-1})|}{|\homtA^{-1}|} \nonumber \\
	= \frac{|\homtA^{-1}||\homdA^{-1}||\text{ker}(\delta_{\tilde{A}}^{-1})|}{|\text{ker}(\delta^{-1}_{\tilde{A}})||\text{ker}(\delta^{-1}_{\partial A})||\homtA^{-1}|} \nonumber \\
	= \frac{|\homdA^{-1}|}{|\text{ker}(\delta_{\partial A}^{-1})|} = |\text{Im}(\delta^{-1}_{\partial A})|, \nonumber
	\end{eqnarray} 
	where we used equations (\ref{eq:hom_dec}), (\ref{eq:ker_dec}) and (\ref{eq:imbA}).

\acknowledgments

JPIJ thanks CNPq (Grant No. 162774/2015-0) for support during this work  and M. I. Jimenez for the valuable discussions.
MP work is supported by CAPES. LNQX thanks CNPq (Grant No.
164523/2018-9) for supporting this work.

\bibliographystyle{JHEP}



%
%

\bibliography{bib2}



\end{document}